\documentclass{article}
\usepackage{kr}

\usepackage{times}
\usepackage{soul}
\usepackage{url}
\usepackage[hidelinks]{hyperref}
\usepackage[utf8]{inputenc}
\usepackage[small]{caption}
\usepackage{graphicx}
\usepackage{amsmath}
\usepackage{amsthm}
\usepackage{booktabs}
\usepackage{algorithm}
\usepackage{algorithmic}
\usepackage{amssymb}
\usepackage{adjustbox}
\usepackage{pdfpages}
\usepackage{upgreek}
\usepackage{latexsym }
\usepackage{ stmaryrd }
\newcommand{\comment}[1]{}
\usepackage{centernot,cancel}
\usepackage[toc,page]{appendix}
\usepackage{tabto}
\usepackage{comment}
\usepackage{multirow}
\usepackage{hyperref}
\usepackage[ligature,reserved,inference]{semantic}
\usepackage{paralist}
\usepackage{xspace}
\usepackage{mathtools}

\usepackage{colortbl}

\usepackage{ltlfonts}
\usepackage{listings}

\newcommand{\Always}{\LTLsquare}
\newcommand{\Event}{\LTLdiamond} 
\newcommand{\Next}{\LTLcircle}

\newcommand{\U}{\mathbin{\mathcal{U}}}

\renewcommand{\And}{\mathrel{\wedge}}
\newcommand{\Or}{\mathrel{\vee}}
\newcommand{\Impl}{\mathrel{\rightarrow}}
\newcommand{\Into}{\Impl}

\newcommand{\LTL}{\ensuremath{\textup{LTL}}\xspace}
\newcommand{\LTLt}{\ensuremath{\textup{LTL}_{\calT}}\xspace}

\newcommand{\phiT}{\ensuremath{\varphi_{\mathcal{T}}}\xspace}
\newcommand{\phiB}{\ensuremath{\varphi_{\mathbb{B}}}\xspace}
\newcommand{\phiExtra}{\varphi^{\textit{extra}}}

\newcommand{\mycal}[1]{\ensuremath{\mathcal{#1}}\xspace}

\newcommand{\calT}{\mycal{T}}

\newcommand{\Boolbb}{\mathbb{B}}

\definecolor{darkGray}{gray}{0.55} 
\definecolor{lightGray}{gray}{0.85} 

\newcommand{\rhoB}{\ensuremath{\rho^{\mathbb{B}}}\xspace}

\newcommand{\Env}{\mathbb{E}}
\newcommand{\Sys}{\mathbb{S}}

\newcommand{\Val}{\ensuremath{\textit{val}}}
\newcommand{\KWDTwo}[1]{\ensuremath{\textit{#1}}\xspace}
\newcommand{\Max}{\KWDTwo{max}}
\newcommand{\Copies}{\KWDTwo{copies}}
\newcommand{\Unroll}{\KWDTwo{unroll}}
\newcommand{\QElim}{\KWDTwo{QElim}}
\newcommand{\Witness}{\KWDTwo{witness}}
\newcommand{\TableBenchmark}{
\begin{table*}[t!]
  \centering
\begin{tabular}{|c|c|c|c|c|c|c|c|c|c|c|c|c|}  \hline
  \multirow{2}{*}{Bn.} & \multirow{2}{*}{Cls.} & \multicolumn{7}{|c|}{Unconditional} & \multicolumn{4}{|c|}{Conditional}  \\ 
 \cline{3-13}
  & & \multicolumn{4}{|c|}{QBF} & \multicolumn{3}{|c|}{Q-SMT} & \multicolumn{2}{|c|}{QBF} & \multicolumn{2}{|c|}{Q-SMT} \\ 
 \cline{3-13}
  (nm.) & (v, l) & Pre. & 10 & 50 & 100 & 10 & 50 & 100 & 10 & 20 & 10 & 20 \\ [0.5ex] 
 \hline\hline
 \multirow{2}{*}{\textit{Li.}} & \cellcolor{darkGray} (5, 16) & $4.41$ & 0.02 & 0.44 & 9.85 & 0.32 & 0.64 & 14.08 & 422 &  - & 611 & - \\ 
  & \cellcolor{lightGray} (5, 16) & $3.71$ & 0.02 & 0.36 & 8.00 & 0.26 & 0.52 & 11.44 & 341 &  - & 493 & - \\

 \hline \hline
  \multirow{2}{*}{\textit{Tr.}} & \cellcolor{darkGray} (19, 36) & 5.13 & 0.01 & 0.53 & 11.70 & 0.38 & 0.76 & 16.72 & 507 &  - & - & - \\
& \cellcolor{lightGray} (19, 36) & 5.01 & 0.01 & 0.57 & 12.62 & 0.41 & 0.82 & 18.04 & 544 &  - & - & - \\

 \hline \hline

 \multirow{2}{*}{\textit{Con.}}& \cellcolor{darkGray} (2, 2) & 4.37 & 0.02 & 0.47 & 10.52 & 0.34 & 0.68 & 15.04 & 452 &  490 & 654 & 7029 \\
  & \cellcolor{lightGray} (2, 2) & 4.34 & 0.02 & 0.51 & 11.39 & 0.37 & 0.74 & 16.28 & 498 &  536 & 741 & - \\ 

 \hline \hline

   \multirow{2}{*}{\textit{Tn.}}& \cellcolor{darkGray} (8, 8) & $7.04$ & 0.02 & 0.43 & 9.54 & 0.31 & 0.62 & 13.64 & 411 &  4454 & - & - \\
  & \cellcolor{lightGray} (8, 8) & $7.38$ & 0.02 & 0.61 & 13.55 & 0.44 & 0.88 & 19.36 & 588 &  6563 & 824 & - \\ 

 \hline \hline

   \multirow{2}{*}{\textit{Air.}}& \cellcolor{darkGray} (13, 14) & $3.29$ & 0.01 & 0.50 & 11.14 & 0.36 & 0.72 & 15.92 & 477 &  - & 6838 & - \\ 
  & \cellcolor{lightGray} (13, 14) & $4.51$ & 0.01 & 0.40 & 8.93 & 0.29 & 0.58 & 12.76 & 397 &  - & 5124 & - \\ 

  \hline \hline
 \multirow{2}{*}{\textit{Coo.}}& \cellcolor{darkGray} (3, 5) & 3.64 & 0.03 & 0.46 & 10.16 & 0.33 & 0.66 & 14.52 & 443 &  475 & 622 & 681 \\
  & \cellcolor{lightGray} (3, 5) & 3.56 & 0.03 & 0.54 & 12.01 & 0.39 & 0.78 & 17.16 & 520 &  556 & 754 & 811 \\

 \hline \hline

\multirow{2}{*}{\textit{Usb}} & \cellcolor{darkGray} (5, 8) & 3.93 & 0.01 & 0.49 & 10.78 & 0.35 & 0.70 & 15.40 & 466 &  5151 & 6557 & - \\ 

& \cellcolor{lightGray} (5, 8) & 3.93& 0.01 & 0.56 & 12.32 & 0.40 & 0.80 & 17.6 & 5343 &  5702 & 7837 & - \\

\hline \hline
  \multirow{2}{*}{\textit{St.}} & \cellcolor{darkGray} (11, 14) & 6.06 & 0.02 & 0.39 & 8.62 & 0.28 & 0.56 & 12.32 & 3417 &  - & - & - \\

 & \cellcolor{lightGray} (11, 14) & 2.86 & 0.02 & 0.51 & 11.39 & 0.37 & 0.74 & 16.28 & 5162 &  - & - & - \\
 \hline

\end{tabular}
\caption{Results using both Q-SMT and QBF unrollings to find unconditional and conditional environment strategies of industrial benchmarks.}
  \label{tab:benchmark}
\end{table*}
}

\newtheorem{example}{Example}
\newtheorem{theorem}{Theorem}
\newtheorem{lemma}{Lemma}
\newtheorem{definition}{Definition}
\newtheorem{remark}{Remark}

\usepackage{tikz,ifthen,pgfplots}
\usetikzlibrary{arrows,trees,backgrounds,automata,shapes,decorations,plotmarks,fit,calc,positioning,shadows,chains}
\usetikzlibrary{quotes,arrows.meta}
\tikzstyle{every pin edge}=[<-,shorten <=1pt]
\tikzstyle{neuron}=[circle,fill=black!25,minimum size=17pt,inner sep=0pt]
\tikzstyle{input neuron}=[neuron, fill=green!50]
\tikzstyle{output neuron}=[neuron, fill=red!50]
\tikzstyle{hidden neuron}=[neuron, fill=blue!50]
\tikzstyle{small neuron}        =[hidden neuron, draw, minimum size=15pt]
\tikzstyle{small input neuron}  =[input neuron , draw, minimum size=15pt]
\tikzstyle{small output neuron} =[output neuron, draw, minimum size=15pt]
\tikzstyle{annot} = [text width=4em, text centered]
\tikzstyle{nnedge} = [-{stealth},shorten >=0.1cm, shorten <=0.05cm,line width=0.8pt,black]
\tikzstyle{edge} = [->,line width = 0.3pt, shorten >=0.2cm]
\tikzstyle{edgeWide} = [->,line width = 2pt, , shorten >=0.2cm]
\usetikzlibrary{calc}

\title{Explanations for Unrealizability of Infinite-State Safety Shields
}

\author{%
Andoni Rodríguez$^{1,2}$\and 
Irfansha Shaik$^{3,4}$\and
Davide Corsi$^5$\and
Roy Fox$^5$\and
César Sánchez$^1$
\affiliations
$^1$IMDEA Software Institute, Spain\\
$^2$Universidad Politécnica de Madrid, Spain\\
$^3$Kvantify Aps, Denmark\\
$^4$Department of Computer Science, Aarhus University, Denmark \\
$^6$University of California Irvine, US\\
}

\begin{document}

\maketitle

\begin{abstract}
Safe Reinforcement Learning focuses on developing optimal policies while ensuring safety.
A popular method to address such task is \textit{shielding},
in which a correct-by-construction safety component is synthetised from
logical specifications.
Recently, shield synthesis has been extended to infinite-state domains,
such as continuous environments.
This makes shielding more applicable to realistic scenarios.
However, often shields might be unrealizable because the specification is inconsistent (e.g., contradictory).
In order to address this gap, we present a method to obtain simple unconditional and conditional explanations that witness unrealizability, which goes by temporal formula unrolling. 
In this paper, we show different variants of the technique and its applicability.

\end{abstract}
\section{Introduction} \label{sec:intro}

Deep Reinforcement Learning (DRL) has been shown to successfully control
reactive systems of high complexity \cite{MaFa20}.
However, despite their success, DRL controllers 
can cause even state-of-the-art agents to react unexpectedly \cite{GoShSz14}.
This issue raises severe concerns regarding the deployment of DRL-based agents in safety-critical reactive systems.
Therefore, techniques that rigorously ensure the safe behaviour of DRL-controlled reactive systems have recently been proposed.
One of them is 
\emph{shielding}~\cite{BlKoKoWa15,AlBloEh18,corsi25efficient}, 
which incorporates an external component (a ``shield'') that \emph{enforces} an agent to behave safely according to a given specification $\varphi$
specified in temporal logic.
In this paper, we use post-shielding, where
the shield does not interrupt the agent unless it violates a safety constraint.

Regularly, shields are built from specifications using 
reactive synthesis
\cite{pnueli89onthesythesis,pnueli89onthesythesis:b}, in which, 
given a specification $\varphi$, a system is crafted that is
guaranteed to satisfy $\varphi$ for all possible behaviors of its
environment.
Realizability is the related decision problem of deciding whether such
a system exists.
This problem is well-studied for specifications written in Linear
temporal logic (LTL) \cite{pnueli77temporal,jacobs17reactive}.
However, many realistic specifications use complex data, whereas 
LTL is inherently propositional.
Alternatively, such realistic specifications can be expressed in LTL modulo theories ($\LTLt$), which replaces
atomic propositions with literals from a first-order theory $\calT$
\cite{geatti22linear}, whose domain might be infinite (e.g., numbers).
In~\cite{rodriguez23boolean,rodriguez24realizability} an $\LTLt$ specification is translated
into an equi-realizable Boolean LTL specification by (1) substituting
theory literals by fresh Boolean variables, and (2) computing, using
theory reasoning, an additional Boolean formula that captures the
dependencies between the new Boolean variables imposed by the
literals.
This approach is called \textit{Boolean abstraction} and paves the way to produce 
infinite-state shields, in which the input and output 
of the shield is no longer Boolean, but belonging to infinite data like numbers \cite{wu19shield,corsi24verification,rodriguez25shield,Kim25realizable}.

However, it is usually the case that the shield is not realizable,
which means that the uncontrollable environment has a way to violate the specification
and the shield no longer provides safety guarantees.
Thus, we want to understand why this happens: i.e, have an explanation.
Recently, the problem of explainability in AI has gained traction,
which has had a major impact on the interest in explaining the intricacies of 
reactive systems \cite{schewe07bounded,baier21causality,bassan23formally}.
These approaches focus on either building simple
reactive systems, tuning their behaviour or explaining, in each timestep, why the system produces an output 
to an environment input.
However, these approaches (1) have not been used for infinite-state systems and
(2) have not addressed the particular problem of understanding shields 
that are to be synthetised from specifications.
Thus, in this paper, we want to open an alternative approach to address explainability: 
to analyze the
$\LTLt$ specification $\varphi$ of a shield itself 
in order to discover why $\varphi$ is unrealizable (in case it is) via
producing a simple witness of an environment strategy such that $\varphi$ is
violated.
\begin{example} \label{ex:intro}
  Consider the following $\LTLt$ specification:
$$\varphi = \Always [(x<2) \shortrightarrow \Next(y>1)] \wedge [(x \geq 2) \shortrightarrow (y < x)],$$ 
where $x,y \in \mathbb{Z}$ and where $x$ is uncontrollable (i.e., belongs to an \textit{environment} player) 
and $y$ is controllable (i.e., belongs to an \textit{environment} player).

\end{example}

Specification $\varphi$ is unrealizable, which means  
that the environment player has a strategy to assign valuations to $x$ such that 
$\varphi$ is violated.
Therefore, it can be argued that we can synthetise (e.g., using \cite{rodriguez24adaptive,rodriguez24predictable}) the environment strategy
for $\varphi$ and analyze it in order
to understand why $\varphi$ is unrealizable.
However, 
this strategy might be complex to interpret and extract simple
explanations (see difficulty in Fig.~\ref{fig:simpleAutomata} of App.~\ref{app:suppl}).
Moreover, this extraction is not automatic and the synthesis procedure might not scale.

%

In this paper, we address these problems.
We propose a general, reusable and an automatic framework to find unrealizability explanations:
a bounded unrealizability checking via
generating finite
\textit{unconditional} and \textit{conditional strategies} as explanations for unrealizable
$\LTLt$ specifications.
The idea with (1) unconditional explanations is to provide a finite set of environment assignments
that shows why a specification is unrealizable.
Alternatively, (2) conditional explanations mean that the environment observes the stateful memory 
in order to take a decision, which means that the explanation is a function over previous moves.
Our general solution works as follows: (1) given $\varphi$, we construct an
\textit{unrolled} formula $\varphi^k$ of length $k$ that creates $k$
copies of variables and instances of $\varphi$; and (2) we use $\varphi^k$
to find witnesses of plays that make the environment player reach
a violation of $\varphi$.
This search is described as a logical formula with different
quantifier alternations, depending on the conditionality of the
strategy sought (meaning that the
environment is \textit{conditioned} to a greater or lesser extent on previous moves).
If the strategy is completely unconditional, then the formula is quantified with
the form $\exists^*\forall^* \textit{. } \varphi$, 
whereas more conditional strategies 
involve some quantifier alternations such as
$\exists^*\forall^*\exists^*\forall^* \textit{. } \varphi$, which can
be understood as split points where the environment \textit{observes}
some prefix of the play.
Our method explores all possible paths of length $k$ and
finds whether there is at least one violation of $\phiT$ 
that
can be achieved by the environment.

The resulting approach is sound, incomplete for unconditional explanations 
and complete up to $k$ for conditional ones.
In summary, the contributions of this paper are as follows:
\begin{enumerate}
\item An approach to solve unconditional strategy search of $\LTLt$ specifications
using satisfiability solvers for quanfidied first-order theory formulae.

\item An alternative method that is based on a combination of a Boolean abstraction, followed by
quantified Boolean formulae and a technique to produce values in $\calT$. 
Unlike the first method, this alternative provides termination guarantees in each iteration of method. 
Both techniques perform an unrolling similar to bounded model
checking~\cite{clarke01bounded} where we use a $k$-length unrolling of
the original specification.

\item An analogous technique to 1-2 to produce conditional explanations.

\item A case study to discuss the tradeoffs between different explanations and 
techniques to obtain them. 

\item Empirical evaluation showcasing the scalability of the approach.
\end{enumerate}

To the best of our knowledge, this is the first work that focuses on explaining \LTLt unrealizability and (infinite-state) shield unrealizability.


%
%
%
%
%
%


\newpage

\section{Preliminaries} \label{sec:prelim}

\subsection{Temporal Logic and Synthesis}
\subsubsection{LTL.}
We consider LTL~\cite{pnueli77temporal,manna95temporal}, whose
formulae contain atomic propositions, $\land$ and $\neg$ (the usual
Boolean conjunction and negation\footnote{We will use $\neg \varphi$
  to represent negation of a formula and $\overline{a}$ to represent
  negation of an atom $a$.}, respectively), and $\Next$ and $\U$ (the
\textit{next} and \textit{until} temporal operators).
The semantics of LTL formulae associates traces
$\sigma\in\Sigma^\omega$ with \LTL formulae as follows: 
\[
  \begin{array}{l@{\hspace{0.3em}}c@{\hspace{0.3em}}l}
    \sigma \models \top && \text{always holds} \\
    \sigma \models a & \text{iff } & a \in\sigma(0) \\
    \sigma \models \varphi_1 \Or \varphi_2 & \text{iff } & \sigma\models \varphi_1 \text{ or } \sigma\models \varphi_2 \\
     \sigma \models \Next \varphi & \text{iff } & \sigma^1\models \varphi \\
     \sigma \models \varphi_1 \U \varphi_2 & \text{iff } & \text{for some } i\geq 0\;\; \sigma^i\models \varphi_2, \text{ and } \\
    && \;\;\;\;\;\text{for all } 
     0\leq j<i, \sigma^j\models\varphi_1 \\
  \end{array}
\]
where
$\sigma \models \top$ always holds and from which we can also derive common operators like $\wedge$ and $\Always$ (which means \textit{always}), etc.
A safety formula $\varphi$ is such that for every failing trace $\sigma\not\models\varphi$ 
there is a finite prefix $u$ of $\sigma$, such that all $\sigma'$ extending $u$ also 
falsify $\varphi$ (i.e. $\sigma'\not\models\varphi$).

\subsubsection{Synthesis.} Reactive synthesis 
is the problem of producing a system from an LTL specification, where
the atomic propositions are split into propositions that are
controlled by the environment and those that are controlled by the
system.
Realizability is the related decision problem of deciding whether such a system exists.
Realizability corresponds to a turn-based game in a finite arena where, in each turn, the
environment produces values of its variables (inputs) and the system
responds with values of its variables (outputs).

We revise now conventions of reactive synthesis research: 
(1) A play is an infinite sequence of turns.
(2) The system player wins a play according to an LTL formula $\varphi$ if
the trace of the play satisfies $\varphi$.
 (3) A strategy $\rho$ of a player is a map from positions into a
move for the player.
(4) A play is played according to $\rho$ if all the moves of the
corresponding player are played according to $\rho$.
Also, (5) $\rho$ is winning for a player if all the possible plays played
according to $\rho$ are winning.
If (6) $\rho$ is winning for the system, the specification is said to be realizable (resp. unrealizable otherwise).

\subsubsection{LTL Modulo Theories.} 

A first-order theory $\calT$ \cite{bradley07calculus} consists of a finite set of functions and constants, 
a set of variables and a domain (which is the sort of its variables). Popular first-order theories are e.g., Presburger arithmetic or real arithmetic.

LTL Modulo Theories ($\LTLt$) is the extension of LTL where propositions are replaced by
literals from a given first-order theory $\calT$ (a finite-trace version is given in \cite{geatti22linear} and a infinite-trace version is given in \cite{rodriguez23boolean}).
The semantics of \LTLt associate traces
$\sigma\in\Sigma_{\calT}^\omega$ with formulae, where for atomic
propositions $\sigma \models l$ holds iff
$\sigma(0) \vDash_{\calT} l$, that is, if the valuation $\sigma(0)$
makes the literal $l$ true.

$\LTLt$ realizability is analogous to LTL realizability, but it corresponds to a game in an
arena where positions may have infinitely many successors
if ranges of variables are infinite.

\newpage

\subsection{Boolean Abstraction for $\LTLt$}

Recently, the field of reactive synthesis beyond Booleans gained traction, and solutions were proposed for different fragments of $\LTLt$; e.g., very new works like \cite{heim25issy,rodriguez25counter,shaun25full}.
In this paper, we want our solution to be complete (up to a bound) and sound. 

Therefore, we build upon \cite{rodriguez23boolean}, which showed that some fragments of reactive $\LTLt$ specifications
can be translated into equi-realizable purely
Boolean LTL specifications via a procedure.
The procedure is called Boolean abstraction or Booleanization, and it
works as follows: given $\varphi$ with literals $l_i$, we get a new
specification $\phiB= \varphi[l_i \leftarrow s_i] \wedge \phiExtra$,
where $s_i$ are fresh Boolean variables, controlled by the system,
that replace the literals.
The additional subformula $\phiExtra$ uses $s_i$ as well as additional
Boolean variables $e_k$ controlled by the environment, and captures
that, for each possible $e_k$, the system has the \textit{power} to select a
response among a collection of choices, where each choice is a truth
valuation of each the variables $s_i$ that represent the literals.
That is, a \textit{choice} is a concrete valuation of the $s_i$
variables (and hence of the literals $l_i$), and a \textit{reaction}
is a collection of choices.
Pairs $(e_k, \bigvee_i(c_i))$ denote that a decision $e_k$ of the
environment can be responded by choosing one of the choices $c_i$ in
the disjunction.
Then, the set of reactions captures precisely the finite collection of
decisions of the environment and the resulting finite responses of the
system.
The set of valid reactions is determined by the literals in
a specific theory 
$\calT$.

\begin{example} \label{ex:running}
  Consider $\varphi$ from Ex.~\ref{ex:intro}, for which the abstraction is as follows:
$ \phiB = (\varphi'' \wedge \Always [(e_0 \leftrightarrow
\overline{e_1}) \wedge (e_0 \leftrightarrow \overline{e_2}) \wedge
(e_1 \leftrightarrow \overline{e_2}) \Into \varphi^\textit{Extra}])$,
where
$\varphi'' = (s_0 \shortrightarrow \Next s_1) \wedge (\overline{s_0}
\shortrightarrow s_2)$ is a direct translation of $\varphi$ ($s_0$
abstracts $(x<2)$, $s_1$ abstracts $(y>1)$ and $s_2$ abstracts
$(y<x)$) that over-approximates the power of the system.
The additional sub-formula $\varphi^{\textit{Extra}}$ corrects the
over-approximation and makes $\varphi_{\mathbb{B}}$ preserve the
original decision power of each player:
\begin{align*}
    \begin{array}{rcl}
      e_{0} & \Into & \big( s_{0} \wedge s_{1} \wedge \overline{ s_{2}}) \vee (s_{0} \wedge \overline{s_{1}} \wedge s_{2}) \vee (s_{0} \wedge \overline{s_{1}} \wedge \overline{s_{2}} \big)
      \\  &\wedge& \\
      e_{1} &\Into &\big(\overline{s_{0}} \wedge s_{1} \wedge \overline{s_{2}}) \vee (\overline{s_{0}} \wedge \overline{s_{1}} \wedge s_{2})
      \\  &\wedge& \\
      e_{2} &\Into &\big(\overline{s_{0}} \wedge s_{1} \wedge s_{2}) \vee (\overline{s_{0}} \wedge s_{1} \wedge \overline{s_{2}}) \vee (\overline{s_{0}} \wedge \overline{s_{1}} \wedge s_{2})
    \end{array}
\end{align*}
%
%
 where $e_0, e_1, e_2$ belong to the environment and $s_0, s_1$ belong
 to the system (and all of them are Boolean).
%
 %
 Intuitively, $e_0$ represents $(x < 2)$, $e_1$ represents $(x = 2)$
 and $e_2$ represents $(x > 2)$.
 Choices are playable valuations in control of the system: $c_0=\{s_0
 \wedge s_1 \wedge s_2\}$, $c_0=\{s_0 \wedge s_1 \wedge
 \overline{s_2}\}$, \ldots $c_7=\{\overline{s_0} \wedge \overline{s_1}
 \wedge \overline{s_2}\}$.
 In other words, the system can respond to $e_0$ with either
 $c_1$, $c_2$ or $c_3$; to $e_1$ with either $c_5$ or
 $c_6$; and to $e_2$ with either $c_4$, $c_5$ or $c_6$.


%
Note that $e_1$ results in a strictly more restrictive set of choices
for the system than $e_2$, which allows the system to choose one more
valuation, specificaly $c_4=(\overline{s_0} \wedge s_1 \wedge s_2 )$.
Thus, a ``clever'' environment will never play $e_2$ and it will play
$e_1$ instead.
Therefore, for simplicity in the paper, we will consider the
simplified (equi-realizable) specification
$\phiB = 
  [(s_{0} \Into s_{1}) \wedge (\overline{s_{0}} \Into s_{2})] \wedge
  \Always [(e_{0} \leftrightarrow \overline{e_{1}}) \Into \varphi_{\textit{Extra}'}]$, where 
$\varphi_{\textit{Extra}'}$ is a version of $\varphi_{\textit{Extra}}$ where $e_2$ is ignored.
%
%
%

%
\end{example}


\newpage

\section{Unconditional Explanations} \label{sec:finding}

\subsection{Unconditional Encoding} \label{subsec:uncondEncod}

The goal of this paper is to generate simple and automatic explanations for
unrealizability in safety $\LTLt$ specifications. 
We soon explain how we perform this automatization, but first introduce a notion of \textit{simplicity}.

Since, in a reactive system, the interplay between the players
can be very complex, it is appealing to find a strategy that reduces the noise of this interplay
as far as possible: in other words, ignoring some moves of the (potentially infinite) interplay in order to get an explanation that is as close as possible to a prefix.
This is particularly interesting in shields, because we provide a simple explanation regardless 
of how sophisticated the shielded policy is, and also regardless of how complex the environment is.
In this paper, we designed these explanations in the form of \textit{unconditional} and \textit{conditional} strategies. Let us begin with the first ones:

\begin{definition}
Given a specification $\varphi$ and a length $k$, we call \textit{unconditional strategies} $\rho_{k}$ 
to strategies in which environment can reach a violation of $\varphi$ in $i$ timesteps with
a sequence of moves that is independent of the system.
\end{definition}

More formally, $\rho_{k}$ is a constant function that assigns valuations to the environment variables up to length $k$.

\begin{example} \label{ex:explainedUncond}
In $\varphi$ of Ex.~\ref{ex:intro} it
suffices for the environment to (1) play a value for $x$ such that $(x<2)$, for
example $x:1$, in timestep $i=0$; and (2) play $x:2$ in $i=1$.
This is a winning strategy for the environment in two steps, no matter
what the system plays.
Thus, we say that the environment has a two-step
\textit{unconditional strategy} (or \textit{explanation}) in $\varphi$ and we denote it
$\rho_{k:2}=\{x^0:1, x^1:2\}$.
\end{example}

Note that, from this point onwards, we use the terms \textit{explanation} and \textit{strategy}
interchangeably, whenever the meaning is clear from the context.
Now, in order to generate unconditional explanations/strategies from $\varphi$ of a length up to $k$, 
we propose to use \textit{unrollings} of $\varphi$ to formalize the existence of this 
statement in some logic.
Unconditional properties are of the form \textit{there is} a sequence
of moves of the environment such that, \textit{for all} moves of the
system, the formula is violated, which corresponds to a prefix
$\exists^*\forall^*$ formula.
%
%

\begin{definition} \label{def:uncondUnroll}
Given a specification $\varphi$ and a length $k$, we call \textit{unconditional unrolling formula}
to an encoding 
$ \psi = \exists a^0,a^1,...,a^{k-1}. \forall
b^0,b^1,...,b^{k-1}. \neg [\varphi_0 \And \varphi_1\And ...\And
\varphi_{k-1}], $ where variables $a_i$ are controlled by the
environment and $b_j$ belong to the system, and 
formulae $\varphi_i$ correspond to instantiations of the
specification $\varphi$ at instant $i$.
\end{definition}

Note that we negate $[\varphi_0 ...]$ in our query, because the objective of the environment is to
find a witness of the negation of $\varphi$.
Also, note that copies $\varphi_i$ resemble classical bounded
model checking~\cite{biere99symbolic,clarke01bounded}.

\begin{remark}
Before the encoding, $\varphi$ is transformed into negation normal form, and then
specialized as $\varphi^{i,k}$ for every step $i$ (and the maximum
unrolling $k$), using the fix-point expansions of the temporal
operators.
When an appropriate solver for $\psi$ is searching to make a formula $\top$, then every
attempt to expand a sub-formula beyond $k$ is replaced by $\bot$ 
(and vice-versa when seeking to make a formula $\bot$).
\end{remark}


\newpage

\subsection{Method \#1: A Q-SMT Encoding} \label{subsec:QSMT}

For a specification $\varphi$ in $\LTLt$,
$\psi$ of Def.~\ref{def:uncondUnroll} is a quantified formula in some first-order theory $\calT$, 
which is a natural encoding that can be solved using quantified satisfiability modulo theories 
(Q-SMT) procedures; for instance, \cite{cooper72theorem} for integers or \cite{collins75quantifier} for reals.

Now, we detail our method and we start with $k=1$ in order to find whether there is an environment
one-step unconditional winning strategy
$\rho_{k=1}$ for $\phiT$.

\begin{example} \label{ex:qsmt}
The $1$-unrolling for $\varphi$ in Ex.~\ref{ex:intro} is: 
$\exists x^0 . \forall y^0. \neg \varphi^{0,1}$, 
where
$$\varphi^{0,1} = [((x^0 < 2) \Into \top) \wedge ((x^0 \geq 2)
\Into (y^0 < x^0))],$$ 
and where $x^0,y^0 \in \mathbb{Z}$ are
variables $x$ and $y$ instantiated in timestep $i=0$.
Note that the sub-formula $\Next(y>1)$ is replaced by $\top$ because
it falls beyond the end of the unrolling, as in standard bounded model
checking (as described in remark 1).
%
%
Thus, the definitive encoding is: 
$$\exists x^0. \forall y^0. \neg (((x^0 < 2) \Into \top)
\wedge ((x^0 \geq 2) \Into (y^0 < x^0))),$$
which is
unsatisfiable (i.e., \texttt{unsat}).
Thus, there is no strategy $\rho_{k=1}$.

However, as seen in Ex.~\ref{ex:explainedUncond}, there exists a two-step unconditional strategy
$\rho_{k=2}$.
%
%
%
%
Moreover, the following Q-SMT encoding (for unrollings depth $k=2$) finds this
witness.
\[\exists x^0,x^1 . \forall y^0,y^1. \neg [\varphi^{0,2}\wedge \varphi^{1,2}], \text{ where } \]
$$\varphi^{0,2} = [((x^0 < 2) \Into (y^1 > 1)) \wedge ((x^0
\geq 2) \Into (y^0 < x^0))]$$ 
and
$$\varphi^{1,2} = [((x^1 < 2) \Into \top) \wedge ((x^1 \geq
2) \Into (y^1 < x^1))]$$
In this case, a Q-SMT solver will respond that the formula is satisfiable
and will output a model such as $m = \{x^0 := 1, x^1 := 2\}$ as an
assignment that satisfies it.
We can see that $m$ is precisely $\rho_{k=2}$.
Most importantly, note how easily the user can understand this explanation comparing to visually inspecting the environment strategy (of Fig.~\ref{fig:simpleAutomata}).

\end{example}
 
In summary, this approach relies on incremental calls to Alg.~\ref{algoMain},
which receives the $\LTLt$ specification $\phiT$, 
an unrolling depth $k=\Max$, 
environment variable set $X$ and system variable set $Y$.
Alg.~\ref{algoMain} also uses sub-procedures:
(1) $\Copies(A,n)$, which performs $n$ timestep copies of the
  set $A$ of variables;
(2) $\Unroll(\varphi,n)$, which performs $m$ unrollings of
  $\varphi$;
(3) $\QElim(X,\varphi)$, which performs quantifier-elimination
  (QE) of set of variables $A$ of variables from formula $\varphi$
  (which must be the inner set of variables); and
(4) $\Witness(\varphi)$, which returns a model of a
  satisfiable formula $\varphi$.
Note that the unrolling formula $F$ (line 5) has
$\{x_0,\ldots,x_n,y_0\ldots,y_n\}$ as free variables, $G$ (line 6)
quantifies universally over $\{y_0\ldots,y_m\}$ so it has
$\{x_0,\ldots,x_n\}$ as free variables, and therefore
$\phiT^{\text{smt}}$ (line 8) is quantifier-free with
$\{x_0,\ldots,x_n\}$ as free variables.
Also, note that, for each of the Q-SMT queries to terminate, we require
$\calT$ to be decidable in the $\exists^*\forall^*$-fragment.
%
%
 
%
\begin{algorithm} [t!]
\begin{algorithmic}[1]
  \REQUIRE $\phiT$, $\Max$, $X$, $Y$ 
  \FOR{$n=1 \text{ to } \Max$}
    \STATE $[y_0,\ldots,y_n] \gets \Copies(Y,n)$ 
    \STATE $[x_0,\ldots,x_n] \gets \Copies(X,n)$ 
    \STATE $F \gets \Unroll(\phiT,n)$ 
    \STATE $G \gets \forall y_0\forall y_1\ldots\forall y_n. F$  
    \STATE $\phiT^{\text{smt}} \gets \QElim([y_0,\ldots,y_n],G)$\\
    \IF{$\neg \phiT^{\text{smt}}$ is SAT}
      \STATE \textbf{return} $(\texttt{true}, \Witness(\phiT^{\text{smt}}))$ 
    \ENDIF
  \ENDFOR
  \STATE \textbf{return} \texttt{uncertain} 
 \caption{Unconditional bounded unrealizability check}
 \label{algoMain}
 \end{algorithmic}
\end{algorithm}

Soundness of Alg.~\ref{algoMain} is due to the following:
\begin{theorem}
\label{thm:soundnessQSMT}
If there is some unrolling depth $k$ such that $\phiT^{\text{smt}}$ is \texttt{sat},
then $\phiT$ is unrealizable.
\end{theorem}

\begin{remark} The Q-SMT encoding described in Ex.~\ref{ex:qsmt} and Alg.~\ref{algoMain} is also suitable for
theories whose satisfiability problems are semi-decidable: if $m$ is
obtained, then it is a \textit{legal} witness of the desired
strategy.
For instance, we can use general-purpose SMT solvers such as Z3 \cite{deMoura08z3} and
encode unrollings for the theory of nonlinear integer
arithmetic.
\end{remark}

\begin{remark} \label{rmk:adaptivity}
Since $\calT$ might have many models satisfying a same valuation of the literals, 
one may be not only interested in
finding a single $\rho$, but instead in
obtaining the \textit{best} strategy with respect to some \textit{soft} criteria.

%
\begin{example} \label{ex:adaptivity}
  Recall the $\rho$ in Ex.~\ref{ex:qsmt} is
  $\rho = \{ (x^0 : 1), (x^1 : 2) \}$.
If the user wants $x_0$ to be as closest as possible to $0$; then an
alternative $\rho' = \{ (x^0 : 0), (x^1 : 2) \}$ is a better
unrealizability witness.
A different criteria is to prioritize \textit{dynamic realism}, 
e.g., if the designers are interested in smooth solutions of the
environment in Ex.~\ref{ex:running}, then they prefer all the environment-controlled 
valuations to be as similar as
possible to each other.
For instance, we would prefer $\rho$ to
$\rho'$, since $x^0:1$ is closer to $x^1:2$ than $x^0:0$.
This
way, designers would like to see, for example, if the environment has a way
to violate the shield specification and without the need of abrupt
movements, but instead more realistic ones. This provides engineers with more
understandable explanations, but heavily depends on the domain of usage 
(further research on this is out of the scope of the paper).
\end{example}

\end{remark}

\begin{algorithm}
\begin{algorithmic}[1]
  \REQUIRE $\phiB$, $\Max$, $E$, $S$ 
  \FOR{$n=1 \text{ to } \Max$}
    \STATE $[s_0,\ldots,s_n] \gets \Copies(S,n)$ 
    \STATE $[e_0,\ldots,e_n] \gets \Copies(E,n)$ 
    \STATE $F \gets \Unroll(\phiB,n)$ 
    \STATE $G \gets \forall s_0\forall s_1\ldots\forall s_n. F$  
    \STATE $\phiB^{\text{qbf}} \gets \QElim([s_0,\ldots,s_n],G)$\\
    \IF{$\neg \phiB^{\text{qbf}}$ is SAT}
      \STATE \textbf{return} $(\texttt{true}, \Witness(\phiB^{\text{qbf}}))$
    \ENDIF
  \ENDFOR
  \STATE \textbf{return} \texttt{uncertain} 
\caption{Uncond. bounded unrealiz. check for LTL.}
 \label{algoQBF}
 \end{algorithmic}
\end{algorithm}


\subsection{Method \#2: Boolean Abstraction to QBF} \label{subsec:qbf}

Although Alg.~\ref{algoMain} is sound (find proofs of theorems in App.~\ref{app:proofs}),
each Q-SMT query lacks termination guarantees \cite{bjorner15playing}.
Therefore, we propose an alternative method based on Boolean abstractions 
for which the evaluation of each unrolling is guaranteed to terminate.
In this method, we generate unrolled formulae in QBF.
This second method follows these steps: (1) we compute an
equi-realizable purely Boolean LTL $\phiB$ from $\varphi$ following the
Boolean abstraction method~\cite{rodriguez23boolean}; (2) find an
unconditional strategy $\rhoB$ in $\phiB$ using QBF solvers; and
(3) translate $\rhoB$ from Booleans to $\calT$ using existential
theory queries, generating a strategy for 
$\varphi$.
%

\begin{example} \label{ex:QBF}
Consider the search for an unconditional
$\rhoB$ for $\phiB$ in Ex.~\ref{ex:running}.
Let
$a_i^{j}$, where $i$ refers to the variable and $j$ to the timestep. Again, $k$ is
the unrolling depth.
We now encode the problem in QBF with unrolling $k=1$:
\[
  \exists e_0^0,e_1^0.\forall s_0^0,s_1^0,s_2^0.\neg \phiB^{0,1},
\]
where $\phiB^{0,1} = [(s_0^0 \Into \top) \wedge (\overline{s_0^0} \Into s_2^0)] 
\Into[(e_0^0 \leftrightarrow e_1^0) \Into \varphi_{\textit{Extra}}^{0,1}]$ and where $\varphi_{\textit{Extra}}^{0,1}$:
\begin{align*}
    \begin{array}{rcl}
      e_0^0 & \Into & \big( s_0^0 \wedge s_1^0 \wedge \overline{s_2^0}) \vee (s_0^0 \wedge \overline{s_1^0} \wedge s_2^0) \vee (s_0^0 \wedge \overline{s_1^0} \wedge \overline{s_2^0} \big)
      \\  &\wedge& \\
      e_1^0 &\Into &\big(\overline{s_0^0} \wedge s_1^0 \wedge \overline{s_2^0}) \vee (\overline{s_0^0} \wedge \overline{s_1^0} \wedge s_2^0)
    \end{array}
\end{align*}
%
%
%
Note, again like in remark 1, that since $k=1$, every attempt to access a proposition
at time step greater than or equal $k$ (for example, $s_1^1$) is
replaced by $\top$.
Hence, the formula $(s_0^0 \Into \top)$ does not impose an actual
constraint.
As we saw in Subsec.~\ref{subsec:QSMT}, there is no strategy
$\rho$ for $\varphi$.
Similarly, the QBF unrolling cannot find a strategy
$\rhoB$ for $\phiB$, as the system only has to satisfy
$(\overline{s_0^0} \Into s_2^0)$, which happens if environment
plays $e_0^0$.
%
This is as expected, since $\varphi$ and $\phiB$ should be
equi-realizable and, thus, strategies should be mutually imitable (in this case, have the same length).

Analogously, since there is $\rho_{k=2}$ for $\varphi$, then
there is a corresponding $\rhoB_{k=2}$ for $\phiB$ via
a $k=2$ QBF unrolling:
\begin{align*}
    \exists e_0^0, & e_1^0, e_0^1,e_1^1. \forall s_0^0,s_1^0,s_2^0,s_0^1,s_1^1,s_2^1.\neg (\phiB^{0,2} \wedge \phiB^{1,2}),
\end{align*}
%
where $\phiB^{0,1} = [(s_0^0 \Into s_1^0) \wedge (\overline{s_0^0} \Into s_2^0)] 
\Into [(e_0^0 \leftrightarrow \neg e_0^1) \Into \varphi_{\textit{Extra}}^{0,1}]$, with $\varphi_{\textit{Extra}}^{0,1}=\varphi_{\textit{Extra}}^{0,2}$ and
where
$\phiB^{1,2} =  [(s_0^1 \Into \top) \wedge (\overline{s_0^1} \Into s_2^1)] \Into [(e_0^1 \leftrightarrow \overline{e_1^1}) \Into
 \varphi_{\textit{Extra}}^{1,2}]$, with $\varphi_{\textit{Extra}}^{1,2}$:
\begin{align*}
    \begin{array}{rcl}
      e_0^1 & \Into & \big( s_0^1 \wedge s_1^1 \wedge \overline{s_2^1}) \vee (s_0^1 \wedge \overline{s_1^1} \wedge s_2^1) \vee (s_0^1 \wedge \overline{s_1^1} \wedge \overline{s_2^1} \big)
      \\  &\wedge& \\
      e_1^1 &\Into &\big(\overline{s_0^1} \wedge s_1^1 \wedge \overline{s_2^1}) \vee (\overline{s_0^1} \wedge \overline{s_1^1} \wedge s_2^1)
    \end{array}
\end{align*}
%
%
%
Note again that an attempt to generate $s_1^2$ is replaced by $\top$,
so the formula $(s_0^1 \Into \top)$ again imposes no constraint.
This time, the environment has a winning strategy if it plays $e_0^0$
and $e_1^1$: playing $e_0^0$ in timestep $i=0$ forces $s_1^1$ to hold
in timestep $i=1$; and $e_1^1$ in $i=1$ forces $s_2^1$ in
$i=1$.
Then, only $(\overline{s_0^1} \wedge s_1^1 \wedge \overline{s_2^1})$ and
$(\overline{s_0^1} \wedge \overline{s_1^1} \wedge s_2^1)$ are valid responses of
the system, but none of them satisfied both $s_2^1$ forced in timestep
$1$ and $s_1^1$ forced by the previous timestep, so the specification
is inevitably violated.
Thus, in $\phiB$, it exists $\rhoB_{k=2}=\{e_0^0, e_1^1\}$.
Last, if we recall that $e_0$ represents $(x < 2)$ and $e_1$ represents
$(x = 2)$, the strategy $\rhoB$ in $\phiB$
obtained with QBF encoding is coherent with
$\rho_{k=2}=\{x^0:1, x^1:2\}$ obtained with Q-SMT 
$\phiT$ in Ex.~\ref{ex:qsmt}.
\end{example}

In summary, this approach starts by $\varphi$, performs abstraction $\phiB$, 
and then again relies on incremental calls to
an algorithm that is identical to Alg.~\ref{algoMain}, except for the fact that it
receives $\phiB$, Boolean environment variable set $E$ and
Boolean system variable set $S$, and does not solve an SMT query, but
a QBF query (see Alg.~\ref{algoQBF}).
Soundness of this method follows analogously to Thm.~\ref{thm:soundnessQSMT},
given that the Boolean abstraction method in use ensures equi-realizability 
of $\varphi$ and $\phiB$, which holds with \cite{rodriguez23boolean}.

\begin{theorem}
\label{thm:soundnessQBF}
If there is some unrolling depth $k$ such that $\phiB^{\text{qbf}}$ is \texttt{sat},
then $\phiB$ is unrealizable.
\end{theorem}

\begin{remark}
Even though the performance of QBF solvers degrades with quantifier
alternations, modern solvers scale efficiently even for large formulae
with $\exists^*\forall^*$ prefixes.
Although scalability is not the goal of methods presented in 
this paper, QBF is a more mature technology than Q-SMT, 
which suggests that performance might also be gained regularly
(see Sec.~\ref{sec:empirical} for experiments in scalability).
\end{remark}

\begin{remark}
Similar to remark 2, note that our method is agnostic to the abstraction method in use
(even semi-decidable methods), 
as long as it preserves equi-realizability.
\end{remark}

\begin{remark}
Similar to remark 3, 
note that $\rhoB$ of Ex.~\ref{ex:QBF}, could also be related to other strategies
in $\phiT$ rather than $\rho$: indeed, 
$\rho'=\{x^0:0, x^1:2\}$, to $\rho''=\{x^0:-1, x^1:2\}$,
to $\rho'''=\{x^0:-2, x^1:2\}$ and an infinite amount of
strategies that make the literal $(x^0<2)$ true.
%
\end{remark}

\begin{remark}
One might wonder when is a Q-SMT encoding preferable to the QBF encoding.
%
%
We believe there are at least three situations to consider: (1) If the abstraction is not terminating, then we can try a Q-SMT encoding. (2) If the theory $\calT$ is undecidable, then still Q-SMT solvers have heuristics for semi-decidability that may produce explanations (which are correct by construction). (3) In very small instances, abstraction might consume most of the time of a \texttt{abstraction+QBF} query, so Q-SMT might be faster. 

\end{remark}
\begin{algorithm} [b!]
\begin{algorithmic}[1]
  \REQUIRE $\phiT$, $\Max$, $X$, $Y$ 
  \FOR{$n=1 \text{ to } \Max$}
    \STATE $[y_0,\ldots,y_n] \gets \Copies(Y,n)$ 
    \STATE $[x_0,\ldots,x_n] \gets \Copies(X,n)$ 
    \STATE $F \gets \Unroll(\phiT,n)$ 
    \STATE $G \gets \texttt{alternats}([y_0,\ldots,y_n],[x_0,\ldots,x_n]). F$  
    \STATE $\phiT^{\text{smt}} \gets \QElim([y_0,\ldots,y_n],G)$\\
    \IF{$\neg \phiT^{\text{smt}}$ is SAT}
      \STATE \textbf{return} $(\texttt{true}, \Witness(\phiT^{\text{smt}}))$ 
    \ENDIF
  \ENDFOR
  \STATE \textbf{return} \texttt{uncertain} 
 \caption{Conditional bounded unrealizability check}
 \label{algoMainPlus}
 \end{algorithmic}
\end{algorithm}


\subsection{Strategy Deabstraction} \label{subsec:deabstraction}

Since the explanation of Alg.~\ref{algoQBF} is given in terms of 
Boolean variables, we need a technique to produce proper values in the domain of
the theory $\calT$. We describe now how to craft a strategy $\rho$ from
$\rhoB$.

\begin{definition}
Consider an $\LTLt$ specification $\varphi$ and its equi-realizable abstraction $\phiB$.
Then, if $\rhoB$ is winning for the environment in $\phiB$, 
we call a \textit{strategy deabstraction} function to a function 
$d: \rhoB \Into \rho$
such that if $\rhoB$ is winning in $\phiB$ then $\rho$ is winning in $\varphi$.
\end{definition}
%

%
%
%

In the case of unconditional strategies, these deabstracted strategies can be obtained leveraging the partitions used
during the Boolean abstraction.

\begin{example} \label{ex:reacts}
Consider Ex.~\ref{ex:running} and the two environment variables 
$e_0$ and $e_1$, which are discrete partitions of the
infinite input space of for variables of the environment.
Recall from Sec.~\ref{sec:prelim} that this partition comes from the fact that the 
abstraction algorithm found reactions
$r_0=(x < 2)$ and $r_1=(x = 2)$ to be valid.
\end{example}

In other words, in order to get a literal in $\calT$ from an environment variable 
$e_k$ in $\phiB$, it suffices to associate each $e_k$ to the 
reaction formula in $\calT$ that is valid.

\begin{definition}
Let function $\texttt{conv}$ produce a formula in $\calT$ 
from each $e_k$.
Then, our deabstraction procedure is: 
$$d_{\rhoB\text{ to }\rho}: \bigwedge_{e_k^i \in \rhoB}^{i \leq k}(\texttt{conv}(e_k^i))$$
\end{definition}

%
%

%
This way, we only need to produce a
$\mathcal{T}$-valuation of $x$ that satisfies the reaction associated to
$e_k$.

\begin{example}
For the reactions in Ex~\ref{ex:reacts}, assignments $\exists x . r_0(x)\gets 1$ and 
$\exists x . r_1(x) \gets 2$ mean
that $e_0$, which represents, $(x < 2)$ has a witness $x:1$; and $e_1$,
which represents $(x = 2)$, has the (only possible) witness $x:2$.
This allows to characterize the strategy of Ex.~\ref{ex:running} as
expected.
Since the Boolean strategy of the environment is
$\rhoB = \{e^0_0, e^1_1\}$, then a possible
$\mathcal{T}$-strategy is $\rho = \{ (x^0 : 1), (x^1 : 2) \}$, 
as predicted in Ex.~\ref{ex:QBF}.
\end{example}

%

We now formalize that any Boolean strategy
$\rho_{\Boolbb} = \{e_i^0, e_j^1,...\}$ of the environment is related
to a $\rho = \{ (x^0 : a), (x^1 : b), ... \}$, where
$a,b \in \calT$, in the sense that the outcomes of $\phiB$ and $\phiT$
are related in terms of literals they satisfy.

\begin{theorem} \label{thm:correspondence}
  Let $\varphi$ and $\phiB$ it Boolean abstraction.  Then, an
  unconditional strategy $\rho$ of length $k$ exists if
  and only if $\rhoB$ of length $k$ exists.
\end{theorem}

\begin{remark} Thm.~\ref{thm:correspondence} is different to correctness theorems of 
\cite{rodriguez23boolean}, because we prove that a strategy of 
\textbf{same length} exists in both $\varphi$ and $\phiB$. 
\end{remark}

\begin{remark} Soundness of deabstraction (and thus Thm.~\ref{thm:correspondence}) is preserved with any model that satisfies a given choice, meaning that, like we did in remark~\ref{rmk:adaptivity} of the Q-SMT method of Sec.~\ref{subsec:QSMT}, we can compute the witness that maximizes a certain soft criteria.
\end{remark}



\section{Conditional Explanations} \label{sec:conditional}

In Sec.~\ref{sec:finding} we showed how to discover unconditional explanations for unrealizability, 
which are simple and preferrable for the solvers (because of the lack of quantifier alternations).
However, 
they carry two fundamental problems: 
(1) incompleteness (unsatisfiability of a k-step unrolling does not mean the formula is unrealizable in $k$ timesteps);
and (2) lack of practical applicability, 
due to the fact that often an unrealizable formula does not contain an unconditional strategy that explains it.
Both problems require that the environment 
movements are no longer independent to the system, but instead has the ability to look at them
in order to choose the move most promising to win.

\begin{definition}
Given a specification $\varphi$ and a length $k$, we call \textit{conditional strategies} $\rho_{k}$ 
to strategies in which environment can reach a violation of $\varphi$ in $i$ timesteps with
a sequence of moves that depends on previous moves of the system.
\end{definition}

\begin{example} \label{ex:conditional}
We illustrate conditionality with $\varphi$':
\begin{align*}
   (x<5) & \rightarrow & \Next^3(y>9) \\
   \wedge \\ 
   (5 \leq x < 10) & \rightarrow & [((y<0) \rightarrow \Next (y>9)) \\ && \wedge ((y \geq 0) \rightarrow \Next (y \leq 9))] \\
   \wedge \\
   (10 \leq x \leq 15) & \rightarrow & (y<x) \\
   (x>15) & \rightarrow & (y>x),
\end{align*}
which is unrealizable.
Again, we can synthetise an environment automata for $\varphi'$
and again the corresponding automata is not easy to interpret 
(see Fig.~\ref{fig:conditionalAutomata} in App.~\ref{app:suppl}).
Instead, we can get both unconditional and conditional explanations.

First, $\varphi'$ is unconditionally unrealizable in 4 timesteps: if the environment plays $x:(x<5)$ in $t_k$, 
then $y>9$ has to hold in $t_{k+3}$; and if $x:10$ in $t_{k+3}$, then $(y<x)$ has to hold (i.e., $y<10$), which contradicts $(y>9)$.
Note that, if the first line of $\varphi'$ did not exist, then there would not exists an unconditional strategy.

Additionally, $\varphi$ does contain a shorter strategy of the environment to win, but this strategy is conditional 
(i.e., looks at the previous play of the system): concretely, 
the environment can play $x:(5 \leq x< 10)$ in $t_k$, then looks at whether the system played $(y<0)$ or $(y\geq0)$ in $t_k$, 
which imposes $(y>0)$ or $(y \leq 9)$ respectively in $t_{k+1}$
and thus the environment reacts accordingly in $k+1$ by playing $(10 \leq x \leq 15)$ or $(x>15)$ respectively
in order to violate $\varphi$. 
Note how this strategy is shorter, but more complex.
\end{example}

\begin{figure}[t]
\centering
\minipage{0.32\linewidth}
    \includegraphics[width=\linewidth]{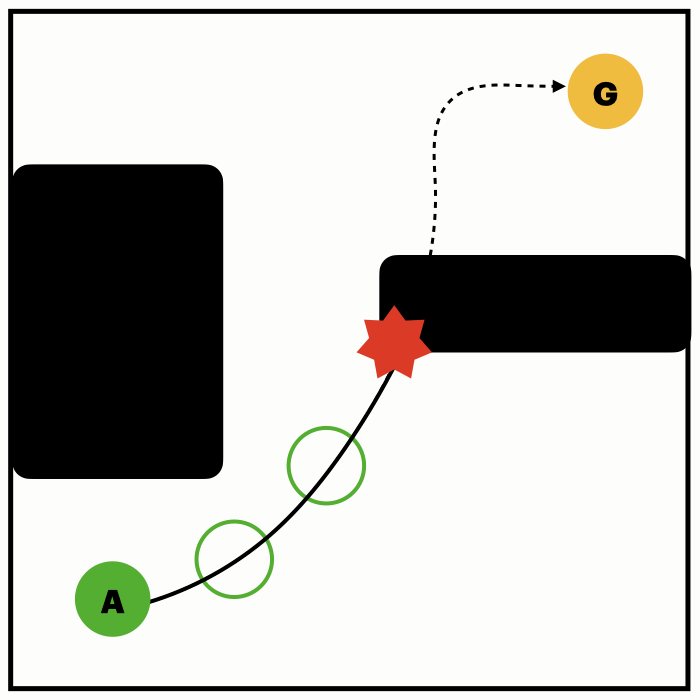}
\endminipage \hfill
\minipage{0.32\linewidth}
    \includegraphics[width=\linewidth]{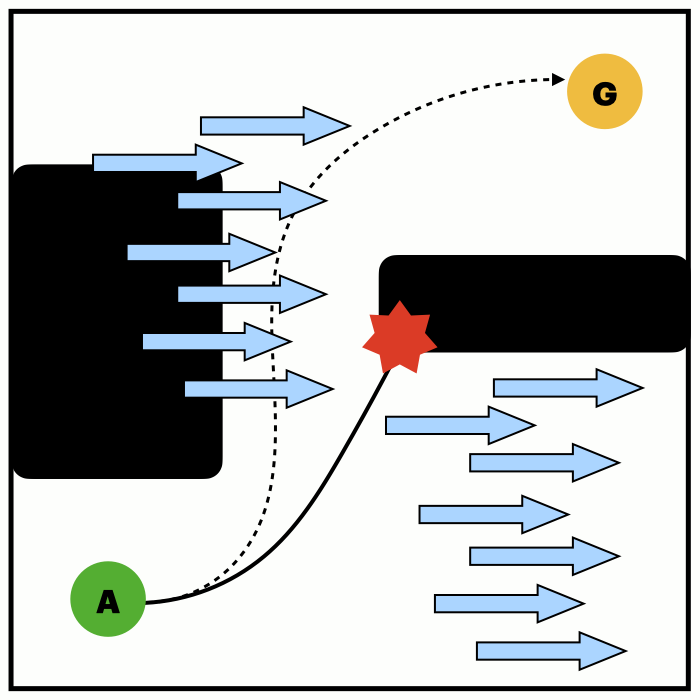}
\endminipage  \hfill
\minipage{0.32\linewidth}
    \includegraphics[width=\linewidth]{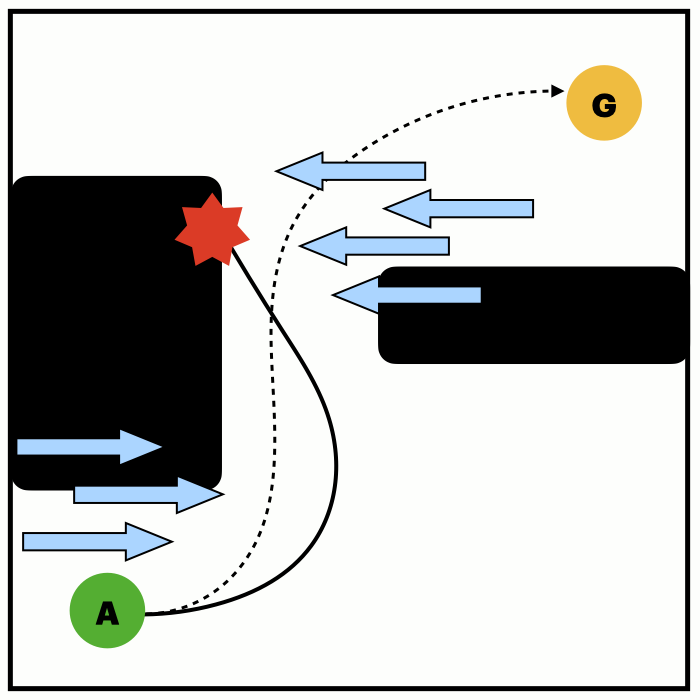}
\endminipage
\caption{Illustration of agent behavior under different conditions.}
\label{fig:motivation:training}
\end{figure}

In order to automatically produce these explanations, we 
now present 
\textit{conditional} strategy search,
for which we need to solve $(\exists^*\forall^*)^*$ formulae.

\begin{definition} \label{def:condUnroll}
Given a specification $\varphi$ and a length $k$, we call \textit{conditional unrolling formula}
to an encoding 
$$\exists a^0. \forall b^0.  \exists a^1. \forall b^1,...,
\exists a^{k-1}. \forall= b^{k-1}. \neg[\varphi_0 \And ...\And
\varphi_{k-1}],$$ where again variables $a_i$ are controlled by the
environment and $b_j$ belong to the system and formulae $\varphi_i$ correspond to instantiations 
of $\varphi$ at instant $i$.
\end{definition}

\begin{example} \label{ex:conditionalSimple}
Consider $\varphi'$ from Ex.~\ref{ex:conditional}.
Unfortunately, due to space limitations, we cannot show the 1-step unrolling
(only note that its verdict is \texttt{unsat} and that both conditional and unconditional 
formulae are the same).
%
Let us denote with Q an arbitrary quantification.
Then, the 2-step unrolling is 
$Q. \neg [\varphi^{0,1} \wedge \varphi^{1,2} ]$, 
where
$\varphi^{0,1}$:
\begin{align*}
   (x^0<5) & \rightarrow & \top \\
   (5 \leq x^0 < 10) & \rightarrow & [((y^0<0) \rightarrow (y^1>9)) \\ && \wedge ((y^0 \geq 0) \rightarrow (y^1 \leq 9))] \\
   (10 \leq x^0 \leq 15) & \rightarrow & (y^0<x) \\
   (x^0>15) & \rightarrow & (y^0>x),
\end{align*}
and where
$\varphi^{1,2}$:
\begin{align*}
   (x^1<5) & \rightarrow & \top \\
   (5 \leq x^1 < 10) & \rightarrow & \top \\ 
   (10 \leq x^1 \leq 15) & \rightarrow & (y^1<x) \\
   (x^1 >15) & \rightarrow & (y^1>x)
\end{align*}
%

As we can see, an unconditional instantiation, 
$Q = \exists x^0 . \forall y^0. \exists x^1 . \forall y^1$, 
results in an \texttt{unsat} verdict 
(because there is no way the environment can win in 2 steps unconditionally, 
although it can do it in 3 steps as shown in Ex.~\ref{ex:conditional}).
However, if the environment is allowed to observe the prefix of the play, then the 
conditional instantiation 
$Q' = \exists x^0 . \forall y^0. \exists x^1 . \forall y^1$ 
is \texttt{sat} with a model $m'$.

Moreover, since a conditional formula captures the environment's ability to to look
at the trace and adapt the behaviour, a model $m'$ is no longer an assignment of environment 
variables to valuations; instead, this only happens in instant $i=0$, 
whereas for any $i>0$ the valuation of $x^i$ is decided using a Skolem function
that depends on the previous value of the system.
For instance, $\{x^0:7, x^1:f_{x^1}(y^0)\}$, where:
\[
  f_{x^1}(y^0) = \begin{cases}
    10 & \text{if $(y^0 < 5)$}\\
    18 & \text{otherwise}
  \end{cases}
\]

\end{example}

In summary, we can construct Alg.~\ref{algoMainPlus} where changes with respect to Alg.~\ref{algoMain} are that
(1) $G$ from line 6 now represents quantifier alternations, 
(2) \texttt{witness} from line 8 does not produce constant Skolem functions
and (3) line 11 returns \texttt{unreal} (see Thm.~\ref{thm:completenessQSMT} below).
Soundness of Alg.~\ref{algoMainPlus} for safety is due to the following:

 \begin{theorem}
 \label{thm:soundnessConditionalQSMT}
If there is some unrolling depth $k$ such that $\phiT^{\text{smt}}$ of Alg.~\ref{algoMainPlus} is \texttt{sat},
then $\phiT$ is unrealizable.
\end{theorem}

Moreover, for arbitrary $k$, Alg.~\ref{algoMainPlus} is guaranteed to find a conditional strategy of length $k$, 
whenever $\varphi$ is an unrealizable safety specifications.
This provides $k$-completeness.

 \begin{theorem}
 \label{thm:completenessQSMT}
 If $\varphi$ is an unrealizable safety formula in $\LTLt$, 
 then there is an unrolling $k$ such that $\phiT^{\text{smt}}$ of Alg.~\ref{algoMainPlus} is \texttt{sat}.
\end{theorem}

However, as in Sec.~\ref{sec:finding}, extracting simple conditional explanations is a challenge and
using Q-SMT for $(\exists^*\forall^*)^*$ can be intractable or
even undecidable for some theories.
On the other side, Boolean abstraction for $\LTLt$ is decidable for $\exists^*\forall^*$ decidable theories, 
so whenever the abstraction is obtained, we can perform a QBF encoding that is analogous to the one in Subsec.~\ref{subsec:qbf}: 
this approach (see Alg.~\ref{algoQBFPlus} in App.~\ref{app:suppl}) starts by $\phiT$, performs abstraction $\phiB$
then again relies on incremental calls to
an algorithm that is identical to Alg.~\ref{algoMainPlus}, except for the fact that it
receives $\phiB$, Boolean environment variable set $E$ and
Boolean system variable set $S$, and does not solve an SMT query, but
a QBF query.
Soundness of and completeness of this method follow analogously.

\begin{theorem}
\label{thm:soundnessConditionalQBF}
If there is some unrolling depth $k$ such that conditional $\phiB^{\text{qbf}}$ is \texttt{sat},
then $\phiB$ is unrealizable.
\end{theorem}

 \begin{theorem}
 \label{thm:completenessConditionalQBF}
 If $\varphi$ is an unrealizable safety formula in $\LTLt$, 
 then there is an unrolling $k$ such that cond. $\phiB^{\text{qbf}}$ is \texttt{sat}.
\end{theorem}

However, again we need a deabstraction procedure, because the explanations that we will get
will be made up of Boolean Skolem functions.
Moreover, leaning on Subsec.~\ref{subsec:deabstraction}, the deabstraction procedure for the Boolean Skolem function is trivial:
(1) in timestep $i=0$ it substitutes the environment assignments by reactions (exactly as in the unconditional case);
whereas (2) in $i>$ it takes each Boolean Skolem function and performs a term substitution by
replacing the domain (variables $s^i$ by their literals in $\calT$) 
and the co-domain (variables $e^i$ by the reactions).
Thus:

 \begin{theorem}
 \label{thm:correspondenceCond}
  Let $\varphi$ and abstraction $\phiB$.  Then, an
  conditional strategy $\rho$ of length $k$ exists iff $\rhoB$ of length $k$ exists.
\end{theorem}

\begin{remark}
We can construct arbitrary versions of the algorithms presented
in this paper if we consider \textit{semi-conditional}
strategies
that are between conditional and unconditional.
For instance, for $k=10$, it can be the case that 
an unconditional unrolling 
(which has $a=1$ quantifier alternations) is
\texttt{unsat}, a conditional encoding
(for which $a=10$) is \texttt{sat} and 
some
semi-conditional encoding (for which $1< a < 10$) is \texttt{sat}.
\end{remark}

\begin{remark}
It is very important to note that a conditional encoding is not necessarily harder to solve than unconditional encodings for QBF solvers.
This is not surprising, since the number of quantifier alternations are not the sole deciding factors for hardness of a problem. For instance, in this paper, we presented Ex.~\ref{ex:intro} which can be solved using unconditional strategy. However, one could easily encode the same problem using a conditional strategy, but this does not increase the intrinsic hardness of the problem.

Indeed, every problem has its intrinsic hardness and it is not trivial to figure out at which layer such a problem belongs to in the polynomial hierarchy. 

\end{remark}

\section{Discussion: What is a \textit{better} explanation?} \label{sec:discussion}

%
Since the unrealizability explanations that we find are the shortest possible 
(because the bounded search ends as soon as a satisfiable assignment is found)
it is easy to see that a shorter explanation 
is better than a longer one.
However, there are other criteria for measuring quality of explanations: 
e.g., if we consider \textit{conditionality}, then 
finding a strategy that is less conditional is better than a more conditional one, 
because the less conditional the strategy, more of its Skolem functions will be constants.
We now formalize this intuition.

\begin{definition}
We measure the quality of an explanation 
as a function $q(k,o)$, where 
$k$ is the length of the explanation and $o$ is the amount of times the environment observes
some prefix of the play (note $o=0$ corresponds to unconditional).
\end{definition}

Since the goal is to minimize $q$, 
then, given $k$ and $o$ for strategy $\rho$ and $k'$ and $o'$ for explanation
$\rho'$, where $k \leq j'$ and $o \leq o'$, then it is easy to see that $q(k,o) \leq q(k',o')$, 
which means $q$ is a \textit{better} explanation.
%
%
However, what if $k \leq k'$ but $o > o'$, or vice-versa? 
For instance, in Ex.~\ref{ex:conditionalSimple}, is it better to have a longer unconditional explanation or a shorter conditional one?
This is not always an easy choice and its automatization will heavily depend on our definition of $q$.
Moreover, both kinds of explanations can be complementary, 
allowing us to find different unrealizability sources.

\begin{example} \label{ex:shield}
Consider a particle scenario where a drone wants to reach a certain goal 
while avoiding dangerous zones (please, refer to App.~\ref{app:drone} for a complete description). 
Also, there is an uncontrollable (environment)
wind-turbulence that can affect the position of the drone. 
Now, the developers (1) train the drone using RL until they reach a desired
success rate in satisfying the goal, and afterwards (2) 
design the specification $\varphi$ of an infinite-state safety shield that would achieve collision avoidance.
However, they find that $\varphi$ is unrealizable.
Moreover, the environment strategy is not easy to interpret (the automata is too big), so they use 
methods of Sec.~\ref{sec:finding} and Sec.~\ref{sec:conditional}.
In particular, they find that there are both conditional and unconditional explanations.

We illustrate this using Fig.~\ref{fig:motivation:training}, where the green dot (A) represents the agent, and the yellow dot (G) represents the goal. 
Also, the dotted line indicates the agent’s intended trajectory, while the solid line represents the actual trajectory executed. 
Now, in Fig.~\ref{fig:motivation:training} \textbf{(Left)} the agent collides without external influences.  
\textbf{(Middle)} The agent attempts to follow a safe trajectory, but unavoidable external perturbations
(the wind represented by blue arrows) push it towards the obstacle, resulting in a collision. 
\textbf{(Right)} The agent initially follows a safe trajectory but is influenced by the wind,
causing a deviation from its intended path. 
Despite this disturbance, the agent successfully recovers and continues towards the goal. 
However, later, a stronger perturbation pushes the agent in the opposite direction, resulting in a collision. 
As we can see, both the unconditional and conditional strategies (i.e., \textbf{Middle}  and \textbf{Right} resp.) 
allow us to discover information about $\varphi$
to restrict the power of the environment towards a realizable version of $\varphi$: \textbf{Middle}
shows that no matter what the agent does, the turbulence is too strong; 
whereas \textbf{Right} shows that $\varphi$ is allowing an unrealistic dynamic such as changing 
the turbulence direction from timestep to timestep.

\end{example}

\TableBenchmark


\section{Empirical Evaluation}
\label{sec:empirical}

\subsubsection{Experimental setting.} The main contribution of this paper is the method for 
obtaining simple explanations in reactive systems specified in $\LTLt$.
In addition, Sec.~\ref{sec:discussion} showcases in a qualitative manner how the method helps finding such explanations for shielding.
Now, we perform a scalability evaluation with a prototype \texttt{unrealExplainerT} that takes an $\LTLt$ formula and a bound $k$, produces the unrolling for both Q-SMT and QBF approaches, and returns the unrealizability explanation (in case it exists). 

Concretely, we created two versions for each benchmark in \cite{rodriguez23boolean}: (1) if the specification
does not have a strategy of length $k$ (e.g., it is
realizable), we introduced
minimal modifications to have such a strategy; (2) if the
specification did contain such strategy, then we made minimal
modifications not to have the strategy.
We tested both the original and modified versions, both QBF and Q-SMT and both conditional and unconditional cases.

\subsubsection{Unconditional encoding.}
We can see results in Tab.~\ref{tab:benchmark}, under the block \textit{unconditional}.
The first column corresponds to the name of the benchmarks (\textit{nm.}) 
The two next columns shows variables (\textit{vr.}) and literals
(\textit{lt.})  per benchmark, where gray colour indicates that the specification is unrealizable and white colour means the
opposite.; 
%
The following two groups of columns show results of the execution of
QBF and Q-SMT techniques for a different (at most) number of unrollings of the
formula: $10$, $50$ and $100$
\footnote{Note that an interval of 100 timesteps in industrial benchmarks happens
because time of the original specifications is dense: this means that if
there is a requirement that must be satisfied in 0.2 seconds and another one in 1 second,
this imposes a discrete representation with 5 timestep horizon. Thus, if another
one must hold for 20 secs, this already has to be represented with 100 timesteps.}.
We show the time needed (in seconds) for each execution, where
Preprocessing time (\textit{Pre.}) in the QBF column is the time needed
for computing the Boolean abstraction.

Results show that,
even if QBF has
an initial Boolean abstraction cost, it quickly begins to to perform better than Q-SMT in time and also reaches higher limits.
Note that we can encounter false negatives (i.e., the specification is unrealizable but because there is no unconditional strategy), but all the results are sound.
Also, note that the time necessary for constructing the unrollings of both QBF and
Q-SMT formulae is negigible for these experiments.

\subsubsection{Conditional encoding.}

Similar experiments were conducted for conditional versions of the benchmarks, noting that this method offers completeness, but at the cost of an expected dramatical decrease in scalability.
The \textit{conditional} block of Tab.~\ref{tab:benchmark} confirms these results, were times increased a lot and Q-SMT usually does not go beyond 5 whereas QBF does not go beyond 20 (usually not beyond 10).
Therefore, we can argue that, in case conditional explanations are seeked, paying the price for the abstraction is absolutely worth it.
Moreover, usually 
Q-SMT could not provide an
answer (underlying SMT-solver responded \texttt{N/A}).


\section{Final Remarks} \label{sec:conclusion}

\subsubsection{Related Work.}
%

We classify our work in classic incremental SAT/SMT/QBF solving methods like 
CEGAR-loops \cite{clarke00counterexample}, which
repeatedly check satisfiability while incrementally growing the
formula/constraint set and bounded unroll constraints each iteration
until the proof is found.
Some concrete works are intended to find counter examples and counter
traces for SMT solvers (e.g.,
\cite{chehida21smt,reynolds15counterexample}).
Also, QBF solvers have been used for such tasks (e.g.,
\cite{balabanov15efficient,janota16solving}) and
\cite{hecking18solving} encodes Petri games in QBF and provides strategies of the
environment.
QBF has also been used for planning
\cite{shaik22classical,shaik2023validation}.

Note that our fully conditional method, is not the same as bounded synthesis \cite{schewe07bounded}, because the latter bound strategies by size of the controller, not by a temporal size. 
However, our technique is similar to bounded model
checking, BMC, ~\cite{biere99symbolic,clarke01bounded} in the sense that BMC 
algorithms unroll a finite-state automata for a fixed
number of steps $k$, and check whether a property violation can occur
in $k$ or fewer steps.
This typically involves encoding the restricted model as an instance
of SAT. The process can be repeated with larger and larger values of
$k$.
Recently, QBF solving has been used in bounded model checking for
hyperproperties~\cite{hsu21bounded}, which opens the door for the same
study over $\LTLt$ properties (and also with loop conditions
\cite{hsu23efficient}).
Last, we find similar research directions in other temporal
logics.; e.g., in STL \cite{maler04monitoring}, the problem of
falsification (i.e., achieving the violation of the
proposed requirements) is considered, via different approaches (e.g.,
\cite{peltomaki22falsification}).

\subsubsection{Limitations and opportunities.}
In this work, we are restricted to the fragment of $\LTLt$ for which synthesis is decidable. Therefore, future work includes proposing similar methods to temporal logics that allow to transfer richer data across time (e.g., infinite-trace \cite{geatti22linear}).
Also, we want to study usability of semi-conditional explanations (see remark 10), because they offer a trade-off between explainability, completeness and performance (e.g., to solve failures of Tab.~\ref{tab:benchmark}) that seems promising.

Last, we consider that the discussion of Sec.~\ref{sec:discussion} about optimizing different explanation criteria has to be widely made (and note that neurosymbolic approaches seem to be a key direction in order to optimize an eventual explainability function $q$).
Thus, we want to integrate our solution under unified views of explainability. 
One example is comparing our work with 
conditional conformant planning \cite{smith98conformant,cimatti00conformant},
where 
the objective is to find a sequence of actions that will guide the system to the desired state, 
regardless of the nondeterminism.

\subsubsection{Conclusion.}
In this paper, we showed methods to find simple explanations of unrealizable $\LTLt$ safety specifications, mostly designed for shields.
We obtain such witnesses via unrollings
of the formula up to a certain number $k$ in a semi-complete
spirit.
We first showed that this method can be designed following an
SMT-with-quantifiers (Q-SMT) encoding, but that this method lacks some
termination guarantees and might not scale.
Then, we showed a second method that uses an QBF encoding preceded by
an abstraction process and a posterior deabstraction process.
In both cases, we proposed a method to generate unconditional explanations (simpler, but less common) and conditional explanations (more complex, but complete).
The paper is the basis for exciting work in other aforementioned directions.


\newpage

\bibliographystyle{kr}
\bibliography{kr-sample}

\begin{thebibliography}{}

\bibitem[\protect\citeauthoryear{Alshiekh \bgroup et al\mbox.\egroup
  }{2018}]{AlBloEh18}
Alshiekh, M.; Bloem, R.; Ehlers, R.; K{\"{o}}nighofer, B.; Niekum, S.; and
  Topcu, U.
\newblock 2018.
\newblock {Safe Reinforcement Learning via Shielding}.
\newblock In {\em Proc. of the 32nd AAAI Conference on Artificial
  Intelligence},  2669--2678.

\bibitem[\protect\citeauthoryear{Azzopardi \bgroup et al\mbox.\egroup
  }{2025}]{shaun25full}
Azzopardi, S.; Stefano, L.~D.; Piterman, N.; and Schneider, G.
\newblock 2025.
\newblock Full ltl synthesis over infinite-state arenas.
\newblock In {\em Proc. of the 37th International Conference on Computer Aided
  Verification (CAV'25)}, LNCS.

\bibitem[\protect\citeauthoryear{Baier \bgroup et al\mbox.\egroup
  }{2021}]{baier21causality}
Baier, C.; Coenen, N.; Finkbeiner, B.; Funke, F.; Jantsch, S.; and Siber, J.
\newblock 2021.
\newblock Causality-based game solving.
\newblock In {\em Proc. of the 33rd International Conference in Computer Aided
  Verification ({CAV}'21), Part {I}}, volume 12759 of {\em LNCS},  894--917.
\newblock Springer.

\bibitem[\protect\citeauthoryear{Balabanov \bgroup et al\mbox.\egroup
  }{2015}]{balabanov15efficient}
Balabanov, V.; Jiang, J.~R.; Janota, M.; and Widl, M.
\newblock 2015.
\newblock Efficient extraction of {QBF} (counter)models from long-distance
  resolution proofs.
\newblock In {\em Proc. of the 29th Conference on Artificial Intelligence
  ({AAAI} 2015), January 25-30, 2015, Austin, Texas, {USA}},  3694--3701.
\newblock {AAAI} Press.

\bibitem[\protect\citeauthoryear{Barrett and
  Tinelli}{2018}]{barrett18satisfiability}
Barrett, C.~W., and Tinelli, C.
\newblock 2018.
\newblock Satisfiability modulo theories.
\newblock In {\em Handbook of Model Checking}. Springer.
\newblock  305--343.

\bibitem[\protect\citeauthoryear{Bassan \bgroup et al\mbox.\egroup
  }{2023}]{bassan23formally}
Bassan, S.; Amir, G.; Corsi, D.; Refaeli, I.; and Katz, G.
\newblock 2023.
\newblock Formally explaining neural networks within reactive systems.
\newblock In {\em Formal Methods in Computer-Aided Design ({FMCAD} 2023)},
  1--13.
\newblock {IEEE}.

\bibitem[\protect\citeauthoryear{Biere \bgroup et al\mbox.\egroup
  }{1999}]{biere99symbolic}
Biere, A.; Cimatti, A.; Clarke, E.~M.; and Zhu, Y.
\newblock 1999.
\newblock Symbolic model checking without {BDDs}.
\newblock In {\em Proc. of the 5th Int'l Confe. on Tools and Algorithms for
  Construction and Analysis of Systems (TACAS'99)}, volume 1579 of {\em LNCS},
  193--207.
\newblock Springer.

\bibitem[\protect\citeauthoryear{Bj{\o}rner and
  Janota}{2015}]{bjorner15playing}
Bj{\o}rner, N.~S., and Janota, M.
\newblock 2015.
\newblock Playing with quantified satisfaction.
\newblock In {\em Proc. of the 20th International Conferences on Logic for
  Programming, Artificial Intelligence and Reasoning (LPAR 2015), Short
  Presentations, Suva, Fiji, November 24-28, 2015}, volume~35 of {\em EPiC
  Series in Computing},  15--27.
\newblock EasyChair.

\bibitem[\protect\citeauthoryear{Bloem \bgroup et al\mbox.\egroup
  }{2015}]{BlKoKoWa15}
Bloem, R.; K{\"{o}}nighofer, B.; K{\"{o}}nighofer, R.; and Wang, C.
\newblock 2015.
\newblock {Shield Synthesis: - Runtime Enforcement for Reactive Systems}.
\newblock In {\em Proc. of the 21st Int. Conf. in Tools and Algorithms for the
  Construction and Analysis of Systems, (TACAS)}, volume 9035,  533--548.

\bibitem[\protect\citeauthoryear{Bradley and Manna}{2007}]{bradley07calculus}
Bradley, A.~R., and Manna, Z.
\newblock 2007.
\newblock {\em The Calculus of Computation}.
\newblock Springer-Verlag.

\bibitem[\protect\citeauthoryear{Chehida \bgroup et al\mbox.\egroup
  }{2021}]{chehida21smt}
Chehida, S.; Ledru, Y.; Blein, Y.; and Vega, G.
\newblock 2021.
\newblock An {SMT}-based approach for generating trace examples and
  counter-examples of parametric properties.
\newblock {\em Int. J. Crit. Comput. Based Syst.} 10(2):143--183.

\bibitem[\protect\citeauthoryear{Cimatti and
  Roveri}{2000}]{cimatti00conformant}
Cimatti, A., and Roveri, M.
\newblock 2000.
\newblock Conformant planning via symbolic model checking.
\newblock {\em J. Artif. Intell. Res.} 13:305--338.

\bibitem[\protect\citeauthoryear{Clarke \bgroup et al\mbox.\egroup
  }{2000}]{clarke00counterexample}
Clarke, E.~M.; Grumberg, O.; Jha, S.; Lu, Y.; and Veith, H.
\newblock 2000.
\newblock Counterexample-guided abstraction refinement.
\newblock In {\em Computer Aided Verification, 12th International Conference,
  {CAV} 2000, Chicago, IL, USA, July 15-19, 2000, Proceedings}, volume 1855 of
  {\em Lecture Notes in Computer Science},  154--169.
\newblock Springer.

\bibitem[\protect\citeauthoryear{Clarke \bgroup et al\mbox.\egroup
  }{2001}]{clarke01bounded}
Clarke, E.~M.; Biere, A.; Raimi, R.; and Zhu, Y.
\newblock 2001.
\newblock Bounded model checking using satisfiability solving.
\newblock {\em Formal Methods Syst. Des.} 19(1):7--34.

\bibitem[\protect\citeauthoryear{Collins}{1975}]{collins75quantifier}
Collins, G.~E.
\newblock 1975.
\newblock Quantifier elimination for real closed fields by cylindrical
  algebraic decompostion.
\newblock In {\em Automata Theory and Formal Languages}, volume~33 of {\em
  LNCS},  134--183.
\newblock Berlin, Heidelberg: Springer.

\bibitem[\protect\citeauthoryear{Cooper}{1972}]{cooper72theorem}
Cooper, D.~W.
\newblock 1972.
\newblock Theorem proving in arithmetic without multiplication.
\newblock {\em Machine Intelligence} 7(2):91--100.

\bibitem[\protect\citeauthoryear{Corsi \bgroup et al\mbox.\egroup
  }{2024}]{corsi24verification}
Corsi, D.; Amir, G.; Rodr{\'{\i}}guez, A.; Katz, G.; S{\'{a}}nchez, C.; and
  Fox, R.
\newblock 2024.
\newblock Verification-guided shielding for deep reinforcement learning.
\newblock {\em {RLJ}} 4:1759--1780.

\bibitem[\protect\citeauthoryear{Corsi \bgroup et al\mbox.\egroup
  }{2025}]{corsi25efficient}
Corsi, D.; Mallik, K.; Rodr{\'{\i}}guez, A.; and S{\'{a}}nchez, C.
\newblock 2025.
\newblock Efficient dynamic shielding for parametric safety specifications.
\newblock In {\em Proc. of the 23rd International Symposium on Automated
  Technology for Verification and Analysis ({ATVA} 2025),}, LNCS.
\newblock Springer.

\bibitem[\protect\citeauthoryear{de Moura and Bj{\o}rner}{2008}]{deMoura08z3}
de~Moura, L.~M., and Bj{\o}rner, N.~S.
\newblock 2008.
\newblock {Z3:} an efficient {SMT} solver.
\newblock In {\em Tools and Algorithms for the Construction and Analysis of
  Systems, 14th International Conference, {TACAS} 2008, Held as Part of the
  Joint European Conferences on Theory and Practice of Software, {ETAPS} 2008,
  Budapest, Hungary, March 29-April 6, 2008. Proceedings}, volume 4963 of {\em
  Lecture Notes in Computer Science},  337--340.
\newblock Springer.

\bibitem[\protect\citeauthoryear{Geatti, Gianola, and
  Gigante}{2022}]{geatti22linear}
Geatti, L.; Gianola, A.; and Gigante, N.
\newblock 2022.
\newblock Linear temporal logic modulo theories over finite traces.
\newblock In {\em Proc. of the 31st International Joint Conference on
  Artificial Intelligence, ({IJCAI} 2022)},  2641--2647.
\newblock ijcai.org.

\bibitem[\protect\citeauthoryear{Goodfellow, Shlens, and
  Szegedy}{2014}]{GoShSz14}
Goodfellow, I.; Shlens, J.; and Szegedy, C.
\newblock 2014.
\newblock {Explaining and Harnessing Adversarial Examples}.
\newblock Technical Report. \url{http://arxiv.org/abs/1412.6572}.

\bibitem[\protect\citeauthoryear{Hecking{-}Harbusch and
  Tentrup}{2018}]{hecking18solving}
Hecking{-}Harbusch, J., and Tentrup, L.
\newblock 2018.
\newblock Solving {QBF} by abstraction.
\newblock In {\em Proc. of the 9th International Symposium on Games, Automata,
  Logics, and Formal Verification, ({GandALF} 2018), Saarbr{\"{u}}cken,
  Germany, 26-28th September 2018}, volume 277 of {\em {EPTCS}},  88--102.

\bibitem[\protect\citeauthoryear{Heim and Dimitrova}{2025}]{heim25issy}
Heim, P., and Dimitrova, R.
\newblock 2025.
\newblock Issy: A comprehensive tool for specification and synthesis of
  infinite-state reactive systems.
\newblock In {\em Proc. of the 37th International Conference on Computer Aided
  Verification (CAV'25)}, LNCS.

\bibitem[\protect\citeauthoryear{Hsu, C{\'{e}}sar~S{\'{a}}nchez, and
  Bonakdarpour}{2023}]{hsu23efficient}
Hsu, T.; C{\'{e}}sar~S{\'{a}}nchez, S.~S.; and Bonakdarpour, B.
\newblock 2023.
\newblock Efficient loop conditions for bounded model checking hyperproperties.
\newblock In {\em Proc. of the 29th International Conference in Tools and
  Algorithms for the Construction and Analysis of Systems (TACAS 2023) , Held
  as Part of the European Joint Conferences on Theory and Practice of Software
  (ETAPS) 2023, Paris, France, April 23 - April 27, 2023}, volume~?? of {\em
  Lecture Notes in Computer Science}, ~??
\newblock Springer.

\bibitem[\protect\citeauthoryear{Hsu, S{\'{a}}nchez, and
  Bonakdarpour}{2021}]{hsu21bounded}
Hsu, T.; S{\'{a}}nchez, C.; and Bonakdarpour, B.
\newblock 2021.
\newblock Bounded model checking for hyperproperties.
\newblock In {\em Proc. of the 27th International Conference in Tools and
  Algorithms for the Construction and Analysis of Systems (TACAS 2021)}, volume
  12651 of {\em Lecture Notes in Computer Science},  94--112.
\newblock Springer.

\bibitem[\protect\citeauthoryear{Jacobs \bgroup et al\mbox.\egroup
  }{2017}]{jacobs17reactive}
Jacobs, S.; Basset, N.; Bloem, R.; Brenguier, R.; Colange, M.; Faymonville, P.;
  Finkbeiner, B.; Khalimov, A.; Klein, F.; Michaud, T.; P{\'{e}}rez, G.~A.;
  Raskin, J.; Sankur, O.; and Tentrup, L.
\newblock 2017.
\newblock The 4th reactive synthesis competition {(SYNTCOMP} 2017): Benchmarks,
  participants {\&} results.
\newblock In {\em Proc. of the 6th Workshop on Synthesis (SYNT@CAV 2017)},
  volume 260 of {\em {EPTCS}},  116--143.

\bibitem[\protect\citeauthoryear{Janota \bgroup et al\mbox.\egroup
  }{2016}]{janota16solving}
Janota, M.; Klieber, W.; Marques{-}Silva, J.; and Clarke, E.~M.
\newblock 2016.
\newblock Solving {QBF} with counterexample guided refinement.
\newblock {\em Artif. Intell.} 234:1--25.

\bibitem[\protect\citeauthoryear{Kim \bgroup et al\mbox.\egroup
  }{2025}]{Kim25realizable}
Kim, K.; Corsi, D.; Rodríguez, A.; Lanier, J.; Parellada, B.; Baldi, P.;
  Sánchez, C.; and Fox, R.
\newblock 2025.
\newblock Realizable continuous-space shields for safe reinforcement learning.
\newblock In {\em Proc. of the 7th Annual Learning for Dynamics {\&} Control
  Conference ({L4DC'25})}, PMLR.

\bibitem[\protect\citeauthoryear{Maler and Nickovic}{2004}]{maler04monitoring}
Maler, O., and Nickovic, D.
\newblock 2004.
\newblock Monitoring temporal properties of continuous signals.
\newblock In {\em Formal Techniques, Modelling and Analysis of Timed and
  Fault-Tolerant Systems, Joint International Conferences on Formal Modelling
  and Analysis of Timed Systems, ({FORMATS} 2004) and Formal Techniques in
  Real-Time and Fault-Tolerant Systems, ({FTRTFT} 2004), Grenoble, France,
  September 22-24, 2004, Proceedings}, volume 3253 of {\em Lecture Notes in
  Computer Science},  152--166.
\newblock Springer.

\bibitem[\protect\citeauthoryear{Manna and Pnueli}{1995}]{manna95temporal}
Manna, Z., and Pnueli, A.
\newblock 1995.
\newblock {\em Temporal verification of reactive systems - safety}.
\newblock Springer.

\bibitem[\protect\citeauthoryear{Marchesini and Farinelli}{2020}]{MaFa20}
Marchesini, E., and Farinelli, A.
\newblock 2020.
\newblock Discrete deep reinforcement learning for mapless navigation.
\newblock In {\em 2020 IEEE International Conference on Robotics and Automation
  (ICRA)}.

\bibitem[\protect\citeauthoryear{Peltom{\"{a}}ki and
  Porres}{2022}]{peltomaki22falsification}
Peltom{\"{a}}ki, J., and Porres, I.
\newblock 2022.
\newblock Falsification of multiple requirements for cyber-physical systems
  using online generative adversarial networks and multi-armed bandits.
\newblock In {\em Proc. of the 15th {IEEE} International Conference on Software
  Testing, Verification and Validation Workshops, ({ICST} Workshops 2022),
  Valencia, Spain, April 4-13, 2022},  21--28.
\newblock {IEEE}.

\bibitem[\protect\citeauthoryear{Pnueli and
  Rosner}{1989a}]{pnueli89onthesythesis:b}
Pnueli, A., and Rosner, R.
\newblock 1989a.
\newblock On the synthesis of a reactive module.
\newblock In {\em Proc. of the 16th Annual ACM Symp. on Principles of
  Programming Languages (POPL'89)},  179--190.
\newblock ACM Press.

\bibitem[\protect\citeauthoryear{Pnueli and
  Rosner}{1989b}]{pnueli89onthesythesis}
Pnueli, A., and Rosner, R.
\newblock 1989b.
\newblock On the synthesis of an asynchronous reactive module.
\newblock In {\em Proc. of the 16th Int'l Colloqium on Automata, Languages and
  Programming (ICALP'89)}, volume 372 of {\em LNCS},  652--671.
\newblock Springer.

\bibitem[\protect\citeauthoryear{{Pnueli}}{1977}]{pnueli77temporal}
{Pnueli}, A.
\newblock 1977.
\newblock The temporal logic of programs.
\newblock {\em Proc. of the 18th Annual Symposium on Foundations of Computer
  Science (FOCS 1977)}  46–57.

\bibitem[\protect\citeauthoryear{Reynolds \bgroup et al\mbox.\egroup
  }{2015}]{reynolds15counterexample}
Reynolds, A.; Deters, M.; Kuncak, V.; Tinelli, C.; and Barrett, C.~W.
\newblock 2015.
\newblock Counterexample-guided quantifier instantiation for synthesis in
  {SMT}.
\newblock In {\em Proc. of the 27th International Conference on Computer Aided
  Verification ({CAV} 2015), San Francisco, CA, USA, July 18-24, 2015}, volume
  9207 of {\em Lecture Notes in Computer Science},  198--216.
\newblock Springer.

\bibitem[\protect\citeauthoryear{Rodr\'iguez and
  S\'anchez}{2023}]{rodriguez23boolean}
Rodr\'iguez, A., and S\'anchez, C.
\newblock 2023.
\newblock {Boolean Abstractions for Realizability Modulo Theories}.
\newblock In {\em Proc. of the 35th International Conference on Computer Aided
  Verification (CAV'23)}, volume 13966 of {\em LNCS}.
\newblock Springer, Cham.

\bibitem[\protect\citeauthoryear{Rodriguez and
  S\'{a}nchez}{2024a}]{rodriguez24adaptive}
Rodriguez, A., and S\'{a}nchez, C.
\newblock 2024a.
\newblock Adaptive reactive synthesis for {LTL} and {LTLf} modulo theories.
\newblock In {\em Proc. of the 38th AAAI Conf. on Artificial Intelligence
  (AAAI'24)}.
\newblock {AAAI} Press.

\bibitem[\protect\citeauthoryear{Rodr{\'{\i}}guez and
  S{\'{a}}nchez}{2024b}]{rodriguez24realizability}
Rodr{\'{\i}}guez, A., and S{\'{a}}nchez, C.
\newblock 2024b.
\newblock Realizability modulo theories.
\newblock {\em J. Log. Algebraic Methods Program. (JLAMP)} 140:100971.

\bibitem[\protect\citeauthoryear{Rodriguez \bgroup et al\mbox.\egroup
  }{2025}]{rodriguez25shield}
Rodriguez, A.; Amir, G.; Corsi, D.; S\'{a}nchez, C.; and Katz, G.
\newblock 2025.
\newblock Shield synthesis for {LTL} modulo theories.
\newblock In {\em Proc. of the 39th AAAI Conf. on Artificial Intelligence
  (AAAI'25)}.
\newblock {AAAI} Press.

\bibitem[\protect\citeauthoryear{Rodr{\'{\i}}guez, Gorostiaga, and
  S{\'{a}}nchez}{2024}]{rodriguez24predictable}
Rodr{\'{\i}}guez, A.; Gorostiaga, F.; and S{\'{a}}nchez, C.
\newblock 2024.
\newblock Predictable and performant reactive synthesis modulo theories via
  functional synthesis.
\newblock In {\em Proc. of the 22nd International Symposium on Automated
  Technology for Verification and Analysis ({ATVA} 2024), Part {II}}, volume
  15055 of {\em LNCS},  28--50.
\newblock Springer.

\bibitem[\protect\citeauthoryear{Rodr\'iguez, Gorostiaga, and
  S\'anchez}{2025}]{rodriguez25counter}
Rodr\'iguez, A.; Gorostiaga, F.; and S\'anchez, C.
\newblock 2025.
\newblock {Counter Example Guided Reactive Synthesis for LTL Modulo Theories}.
\newblock In {\em Proc. of the 37th International Conference on Computer Aided
  Verification (CAV'25)}, LNCS.

\bibitem[\protect\citeauthoryear{Schewe and Finkbeiner}{2007}]{schewe07bounded}
Schewe, S., and Finkbeiner, B.
\newblock 2007.
\newblock Bounded synthesis.
\newblock In {\em Proc. of the 5th International Symposium in Automated
  Technology for Verification and Analysis ({ATVA} 2007)}, volume 4762 of {\em
  LNCS},  474--488.
\newblock Springer.

\bibitem[\protect\citeauthoryear{Shaik and van~de Pol}{2022}]{shaik22classical}
Shaik, I., and van~de Pol, J.
\newblock 2022.
\newblock Classical planning as {QBF} without grounding.
\newblock In {\em Proc. of the 32nd International Conference on Automated
  Planning and Scheduling, ({ICAPS} 2022)},  329--337.
\newblock {AAAI} Press.

\bibitem[\protect\citeauthoryear{Shaik \bgroup et al\mbox.\egroup
  }{2023}]{shaik2023validation}
Shaik, I.; Heisinger, M.; Seidl, M.; and van~de Pol, J.
\newblock 2023.
\newblock Validation of {QBF} encodings with winning strategies.
\newblock In {\em 26th International Conference on Theory and Applications of
  Satisfiability Testing, {SAT} 2023, July 4-8, 2023, Alghero, Italy}, volume
  271 of {\em LIPIcs},  24:1--24:10.
\newblock Schloss Dagstuhl - Leibniz-Zentrum f{\"{u}}r Informatik.

\bibitem[\protect\citeauthoryear{Smith and Weld}{1998}]{smith98conformant}
Smith, D.~E., and Weld, D.~S.
\newblock 1998.
\newblock Conformant graphplan.
\newblock In {\em Proc. of the 15th National Conference on Artificial
  Intelligence ({AAAI} 98)},  889--896.
\newblock {AAAI} Press / The {MIT} Press.

\bibitem[\protect\citeauthoryear{Wu \bgroup et al\mbox.\egroup
  }{2019}]{wu19shield}
Wu, M.; Wang, J.; Deshmukh, J.; and Wang, C.
\newblock 2019.
\newblock Shield synthesis for real: Enforcing safety in cyber-physical
  systems.
\newblock In {\em Proc. of 19th Formal Methods in Computer Aided Design,
  ({FMCAD}'19)},  129--137.
\newblock {IEEE}.

\end{thebibliography}

\newpage

\appendix

\section{First-Order Logic Background}
\label{app:extendedPrelim}

\subsection{Satisfiability}

The (NP-complete) Boolean Satisfiability Problem (SAT) is the decision
problems of determining whether a propositional formula has a
satisfying assignment.
Satisfiability Modulo Theories (SMT)~\cite{barrett18satisfiability}
consists on determining whether a first-order formula is satisfiable
using literals from background theories.
Background theories of interest include arithmetic, arrays,
bit-vectors, inductive data types, uninterpreted functions and
combinations of these.
SMT solving typically deals with existentially quantifier theory
variables, but it has also been extended with capabilities to
handle universal quantifiers, but typically at the expense of
semi-decidability.
In this paper, we call Q-SMT to this quantified SMT extension and
we consider theories for which satisfiability in the
$\exists^*\forall^*$ fragment is decidable (which includes e.g., linear
integer arithmetic and non-linear real arithmetic), because these are the
ones for which Boolean abstraction guarantees termination.

Quantified Boolean formulae (QBF) extend propositional formulae by
allowing arbitrary quantification over the propositional variables.
Unlike SAT and (standard) SMT solving, QBF solving deals with both
existentially and universally quantified Boolean variables.
Unlike SMT, QBF does not consider theory predicates.

\subsection{Skolem Functions}

A Skolem function is a concept that plays a crucial role in the elimination of quantifier alternation in logical formulas.

\begin{definition}
Let \( A(x_1, \ldots, x_n, y) \) be a predicate formula with individual variables \( x_1, \ldots, x_n, y \) whose domains are sets \( X_1, \ldots, X_n, Y \), respectively. A function 
\[
f: X_1 \times \cdots \times X_n \rightarrow Y
\]
is called a Skolem function for the formula 
\[
\exists y \, A(x_1, \ldots, x_n, y)
\]
if, for all \( x_1 \in X_1, \ldots, x_n \in X_n \), the following holds:
\[
\exists y \, A(x_1, \ldots, x_n, y) \Rightarrow A(x_1, \ldots, x_n, f(x_1, \ldots, x_n)).
\]
\end{definition}

Skolem functions are utilized to transform formulas into a form free from the alternation of 
the quantifiers \( \forall \) and \( \exists \). For any formula \( A \) in the language of restricted predicate calculus, 
one can construct a formula in the Skolem normal form:
\[
\exists x_1, \ldots, x_n \, \forall y_1, \ldots, y_m \, C,
\]
where \( C \) does not contain new quantifiers but includes new function symbols. 
The original formula \( A \) is deducible in predicate calculus if and only if its Skolem normal form is.
\newpage

\section{Proofs}
\label{app:proofs}

We now show proof sketched of the theorems in the paper.
Proof of Thm.~\ref{thm:soundnessQSMT} is as follows:

\begin{proof}
  Let $m$ be such that $\phiT^{\text{smt}}$ is SAT and let
  $e_0,\ldots,e_m$ be a model of $\phiT^{\text{smt}}$.
  The unconditional strategy $\rho^{\Env}$ that plays first $e_0$,
  then $e_1$, etc up to $e_m$ is winning for the environment because
  it falsifies $\phiT$.
  Note that the universal quantifiers for $y_0,\ldots,y_m$ guarantee
  that all moves of the system for the first $m$ steps are considered
  and in all cases $\phiT$ is falsified.
  Since $\rho^{\Env}$ is winning, then $\phiT$ is unrealizable.
\end{proof}

Note that Thm.~\ref{thm:soundnessQBF}, Thm.~\ref{thm:soundnessConditionalQSMT} 
and Thm.~\ref{thm:soundnessConditionalQBF} follow analogously.

Proof of Thm.~\ref{thm:correspondence} is as follows, with an auxiliary lemma:

\begin{lemma}[From $e_k$ to $x$]
  \label{lem:e_to_x}
  Let $X=[x_0,x_1, ...]$, $Y=[y_0,y_1, ...]$, $E=[e_0,e_1, ...]$ and $S=[e_0,e_1, ...]$.
  Consider $\phiT(X,Y)$ be an \LTLt formula and $\phiB(E,S)$ be
  its Boolean abstraction.
  Let us denote with $v_a : \Val(A)$, where $A$ is a set of variables.
  For every $e\in E$ there is a computable satisfiable predicate
  $r_e(X)$ such that, for every valuation $v_x\in\Val(X)$ with
  $r_e(v_x)$,
  \begin{itemize}
  \item let $v_y:\Val(A)$ be arbitrary, then the valuation
    $v_s:\Val(S)$ such that $s(s_i)$ is true if and only if $l_i(v_x,v_y)$
    holds $\phiExtra(E,S)$.
  \item let $v_s:\Val(S)$ be such that $\phiExtra(E,S)$ holds, then
    there is a valuation $v_y:\Val(Y)$ that makes $l_i(v_x,v_y)$ hold if
    and only if $s(s_i)$.
  \end{itemize}
  Moreover, for every $v_x:\Val(X)$ there is exactly one $e\in E$
  such that $r_e(v_x)$ holds.
\end{lemma}

\begin{proof}
  The main idea is to create the corresponding sequence using
  Lemma~\ref{lem:e_to_x}.
  It follows that if the system in $\phiT$ can make a collection of
  literals at some point then the system can make the corresponding
  $s_i$ hold at the same point (and viceversa).
  By structural induction, if the atoms have the same valuation then
  all sub-formulae have the same valuation.
  Therefore, given an arbitrary length $l$, if the unconditional strategy
  $\rho_{k=l}$ in $\phiT$ is winning for a player the strategy
  $\rho_{k=l}$ in $\phiB$ is winning for a player as well, and
  vice-versa.
\end{proof}

Note that Thm.~\ref{thm:correspondenceCond} follows analogously.

Proof of Thm.~\ref{thm:completenessQSMT} is as follows:

 \begin{proof}
 Since $\phiT$ is unrealizable, there is some reachability goal
 that the environment satisfies. Moreover, since $\phiT$ is in safety, then
 this goal is reached in a finite number of steps.
 \end{proof}

 Note that Thm.~\ref{thm:completenessConditionalQBF} follows analogously.

\newpage

\section{QBF Encoding}

We now show the QBF encoding for our running example.

\subsection{Unconditional Encoding}


The formula is as follows:

\begin{align*}
&\exists e_0^{0}, e_1^{0}, e_2^{0}, e_3^{0}\\
&\exists e_0^{1}, e_1^{1}, e_2^{1}, e_3^{1}\\
&\exists e_0^{2}, e_1^{2}, e_2^{2}, e_3^{2}\\
&\forall s^{0}, s^{1}, s^{2}\\
&\neg . \\
& \varphi_0: & e_0^{0} \rightarrow s^{2} \\
             && e_1^{0} \rightarrow \neg s^{0} \\
             && e_2^{0} \rightarrow (s^{0} \leftrightarrow s^{1}) \\
             && e_3^{0} \rightarrow s^{0} \\
&\wedge \\
& \varphi_1: & e_0^{1} \rightarrow \top \\
             && e_1^{1} \rightarrow \neg s^{1} \\
             && e_2^{1} \rightarrow (s^{1} \leftrightarrow s^{2}) \\
             && e_3^{1} \rightarrow s^{1} \\
\wedge \\
& \varphi_2: & e_0^{2} \rightarrow \top \\
             && e_1^{2} \rightarrow \neg s^{2} \\
             && e_2^{2} \rightarrow (s^{2} \leftrightarrow \top) \\
             && e_3^{2} \rightarrow s^{2} \\
\end{align*}

In order to solve it, we encode it in the QCIR format:

\begin{verbatim}
exists(1, 2, 3, 4)
exists(5, 6, 7, 8)
exists(9, 10, 11, 12)
forall(13, 14, 15)
output(-40)
#\varphi_0
16 = or(-1 , 15) 
17 = or(-2 , - 13) 
18 = or(13, -14)
19 = or(-13, 14)
20 = and(18, 19)
21 = or(-3, 20) 
22 = or(-4 , 13) 
23 = and(16, 17, 18, 19, 20, 21, 22)
#\varphi_1:
24 = and()
25 = or(-6 , - 14) 
26 = or(14, -15)
27 = or(-14, 15)
28 = and(26,27)
29 = or(-7, 28) 
30 = or(-8 , 14)
31 = and(24, 25, 26, 27, 28, 29, 30)
#\varphi_2:
32 = and() 
33 = or(-10 , - 15) 
34 = or(15)
35 = and()
36 = and(34, 35)
37 = or(-11, 36)
38 = or(-12 , 15)
39 = and(32, 33, 34, 35, 36, 37, 38)
# conjunction:
40 = and(23, 31, 39)
\end{verbatim}

\subsection{Conditional Encoding}

The formula is as follows:

\begin{align*}
&\exists e_0^{0}, e_1^{0}, e_2^{0}, e_3^{0}\\
&\forall s^{0}, \exists s_{c}^{0}\\
&\exists e_0^{1}, e_1^{1}, e_2^{1}, e_3^{1}\\
&\forall s^{1}, \exists s_c^{1}\\
& s^{0} \leftrightarrow s_c^{0}\\
& s^{1} \leftrightarrow s_c^{1}\\
&\neg . \\
& \varphi_0: & e_0^{0} \rightarrow \top \\
             && e_1^{0} \rightarrow \neg s_c^{0} \\
             && e_2^{0} \rightarrow (s_c^{0} \leftrightarrow s_c^{1}) \\
             && e_3^{0} \rightarrow s_c^{0} \\
&\wedge \\
& \varphi_1: & e_0^{1} \rightarrow \top \\
             && e_1^{1} \rightarrow \neg s^{1} \\
             && e_2^{1} \rightarrow (s_c^{1} \leftrightarrow \top) \\
             && e_3^{1} \rightarrow s^{1} \\
\end{align*}

Again, we encode it in QCIR:

\begin{verbatim}
exists(1, 2, 3, 4)
forall(13)
exists(9)
exists(5, 6, 7, 8)
forall(14)
exists(10)
output(55)
#\varphi_0
17 = or(-2 , -13) 
18 = or(13, -14)
19 = or(-13, 14)
20 = and(18, 19)
21 = or(-3, 20) 
22 = or(-4 , 13) 
23 = and(17, 21, 22)
#\varphi_1:
25 = or(-6 , -14) 
29 = or(-7, 14) 
30 = or(-8 , 14)
31 = and(25, 29, 30)
# conjunction:
32 = and(23, 31)
# mutual exclusion, 1,2,3,4:
33 = or(1, 2, 3, 4)
34 = or(-1, -2)
35 = or(-1, -3)
36 = or(-1 , -4)
37 = or(-2, -3)
38 = or(-2, -4)
39 = or(-3, -4)
40 = and(33, 34, 35, 36, 37, 38, 39)
# mutual exclusion, 5,6,7,8:
41 = or(5, 6, 7, 8)
42 = or(-5, -6)
43 = or(-5, -7)
44 = or(-5 , -8)
45 = or(-6, -7)
46 = or(-6, -8)
47 = or(-7, -8)
48 = and(41, 42, 43, 44, 45, 46, 47)
# equality:
49 = or(-9, 13)
50 = or(-13, 9)
51 = and(49, 50)
52 = or(-10, 14)
53 = or(-14, 10)
54 = and(52, 53)
#output:
55 = and(-32, 40, 48, 51, 54)
\end{verbatim}

\subsection{Handling System Moves with existential variables}

Unrealizability can be encoded as a $2$-player turn-based game where environment is the first (existential) player and system is the second (universal) player.
A $2$-player game can be elegantly encoded as a QBF.
For encoding environment, one can use existential variables since we only need a single winning move in each turn.
For System, on the otherhand, we need to encode all possible moves. Thus, we can use universal variables to represent all the moves.
The goal of the environment player is then to \emph{win the game} i.e., make the formula true regardless of the opponent moves.
Now, when translating in a QBF encoding, one cannot disable any turn of universal player.
For example, consider the following formula where $0 \leq i \leq n$:
\begin{equation*}
\exists e_0, \forall s_0, ... \exists e_n, \forall e_n \neg (\varphi \land (e_i \rightarrow \neg s_i))
\end{equation*}
Regardless of the formula $\varphi$, setting $e_i$ to True makes the formula True.
Essentially, the above formula encodes that \emph{If environment plays $e_i$ move, then do not allow the $s_i$ move for the system}.
However, a QBF solver (by construction) will search the complete search space i.e., all possible moves of the universal player.
Disallowing or forcing a move of the universal player simply makes the formula false (or true if negated).
While one can construct a custom QBF solver to disallow or force moves of universal player, we cannot use an off-the-shelf QBF solver.

Alternatively, one can use existential variables for each universal variable to encode the force or disallow the moves on the universal player.
For example, let us extend our previous formula $s_0^{c}, ..., s_n^{c}$ existential variables as follows:
\begin{align}
&\exists e_0, \forall s_0, ... \exists e_n, \forall e_n\\
&\exists s_0^{c}, ..., s_n^{c}\\
& \bigwedge_{k=0}^{n}(s_k \leftrightarrow s_k^{c})\\
&\neg (\varphi \land (e_i \rightarrow \neg s_i^{c}))
\end{align}

Equations 2 and 3 defines existential variables and forces equalities between corresponding universal and existential system variables.
Unlike previous encoding, we allow all possible moves of the universal player.
At the same time, when environment plays $e_i$, we encode that $s_i$ is an invalid move and thus the formula becomes True (by negation) only in that particular move but not others.
The new formula is not trivially true, it now depends on the $\varphi$.
While this technique adds one additional quantifier layer, notices that the search space is not increased.
Due do our equality clauses in the equation 3 by unit clause propagation the search space remains same.
It is also possible to push existential system variables to the layers directly below corresponding universal variables as follows:
\begin{align}
&\exists e_0, \forall s_0, \exists s_0^{c} ... \exists e_n, \forall e_n, \exists s_n^{c}\\
& \bigwedge_{k=0}^{n}(s_k \leftrightarrow s_k^{c})\\
&\neg (\varphi \land (e_i \rightarrow \neg s_i^{c}))
\end{align}
Notice that the number of quantifier layers still only increased by 1.
Similar techniques are used in encoding two-player games as QBF, for instance handling illegal moves \cite{shaik2023validation}.
%


\begin{figure}[t]
\centering
\includegraphics[width=0.6\linewidth]{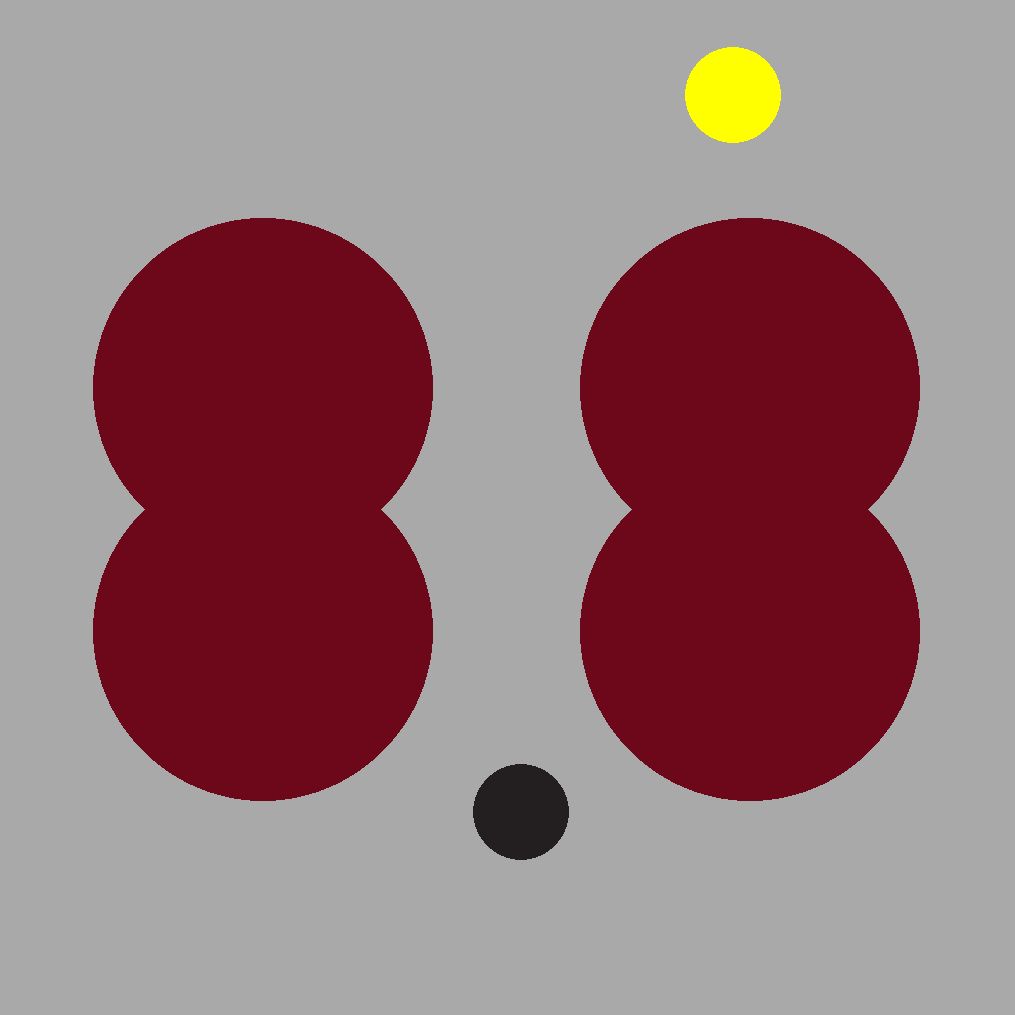}
\caption{Screenshot of the configuration employed for the main set of experiments.} 
\label{fig:appendix:drone2denv}
\end{figure}

\section{The \texttt{Drone2D Environment}} \label{app:drone}

The \textbf{Drone2D Environment} (Fig.\ref{fig:appendix:drone2denv}) is a lightweight continuous control environment where a drone navigates a bounded 2D space to reach a randomly assigned goal position while avoiding collisions with obstacles. The agent's state includes its current position $(x, y)$, velocity $(v_x, v_y)$, and the goal's coordinates $(x_{\text{goal}}, y_{\text{goal}})$. The agent's actions represent accelerations $(a_x, a_y)$ in the $x$- and $y$-directions. The environment simulates realistic dynamics with position and velocity updates, and the agent is rewarded for minimizing its distance to the goal.

\begin{itemize}
    \item \textbf{Unsafe regions} (red circles): The drone cannot cross these zones. A collision with these obstacles results in episode termination with failure.
    \item \textbf{Goal} (yellow circle): The target the agent must reach.
    \item \textbf{Drone} (black circle): The visual representation of the drone.
\end{itemize}

\subsection{State Space}
The state of the drone at time $t$ is represented as:

\[
s_t = [x_t, y_t, x_{\text{goal}}, y_{\text{goal}}, w_x, w_y]
\]

where:
\begin{itemize}
    \item $x_t, y_t$: Current position of the drone.
    \item $x_{\text{goal}}, y_{\text{goal}}$: Coordinates of the goal position.
    \item $w_x, w_y$: Turbulence force affecting the drone’s motion along both axes.
\end{itemize}

The state space is defined as:

\[
x_t, y_t, x_{\text{goal}}, y_{\text{goal}} \in [0, 21], \quad w_x, w_y \in \mathbb{R}
\]

\subsection{Action Space}
The action at time $t$ is represented as:

\[
\mathbf{a}_t = [a_t^x, a_t^y]
\]

where $a_t^x$ and $a_t^y$ are accelerations applied to the drone in the $x$- and $y$-directions, respectively. The actions are continuous and bounded as:

\[
a_t^x, a_t^y \in [-1, 1]
\]

\subsection{Dynamics}
The drone’s position and velocity evolve over discrete time steps according to the following equations:

\[
x_{t+1} = x_t + a_{x, t} \cdot \Delta t + w_x
\]

\[
y_{t+1} = y_t + a_{y, t} \cdot \Delta t + w_y
\]

where the time step is fixed at:

\[
\Delta t = 0.5
\]

\subsection{Sampling Turbulence (Environment Action)}
In this environment, turbulence is represented as an external force affecting the agent's velocity. The turbulence at each time step is defined as a vector \( \mathbf{w}_t = [w_x, w_y] \), where \( w_x \) and \( w_y \) represent the force along the \( x \)- and \( y \)-axes, respectively.

The turbulence is sampled at each time step using the following rule:

\[
w_x, w_y \sim \mathcal{U}(-\sigma, \sigma)
\]

where:
\begin{itemize}
    \item \( \mathcal{U}(-\sigma, \sigma) \) denotes a uniform distribution.
    \item \( \sigma \) is the maximum magnitude of turbulence along each axis (default \( \sigma = 0.6 \)).
\end{itemize}

\subsection{Reward Function}
The reward function encourages the drone to minimize its distance to the target while applying a small constant penalty to incentivize faster completion of the task. The reward at time $t$, denoted as $r_t$, is defined as:

\[
r_t = -(\alpha \cdot d_t) - \beta
\]

where $d_t$ is the Euclidean distance between the current position of the drone $(x_t, y_t)$ and the target position $(x_{\text{goal}}, y_{\text{goal}})$, calculated as:

\[
d_t = \sqrt{(x_t - x_{\text{goal}})^2 + (y_t - y_{\text{goal}})^2}
\]

finally, $\alpha$ is a scaling factor for the distance penalty (set to $\alpha = 0.0005$ by default) and $\beta$ is a small constant penalty to encourage the agent to reach the target quickly (set to $\beta = 0.0001$ by default). This reward function ensures that:
\begin{enumerate}
    \item The agent receives larger penalties for being far from the target ($d_t$).
    \item The small constant penalty ($\beta$) encourages the agent to minimize the number of time steps taken to reach the target.
\end{enumerate}

\subsection{Safety Constraint}
For the main set of experiments, the environment contains four predefined unsafe zones:
\begin{enumerate}
    \item A circle centered at $(5.5, 8.0)$ with a radius of $3.5$.
    \item A circle centered at $(5.5, 13.0)$ with a radius of $3.5$.
    \item A circle centered at $(15.5, 8.0)$ with a radius of $3.5$.
    \item A circle centered at $(15.5, 13.0)$ with a radius of $3.5$.
\end{enumerate}

The union of these regions forms the total unsafe zone $\mathcal{Z}$, defined as $\mathcal{Z} = \{(x, y) \, | \, \sqrt{(x - 5.5)^2 + (y - 8.0)^2} \leq 3.5 \wedge \text{...}\}$\footnote{ Repeated for each obstacle (omitted for calrity).}

Let $\mathbf{1}_{\mathcal{Z}}(x_t, y_t)$ be an indicator function that is $1$ if the agent is in the unsafe zone at time $t$, and $0$ otherwise:

\[
\mathbf{1}_{\mathcal{Z}}(x_t, y_t) =
\begin{cases}
1 & \text{if } (x_t, y_t) \in \mathcal{Z}, \\
0 & \text{otherwise.}
\end{cases}
\]

The agent must satisfy the following constraint:

\[
\sum_{t' = t}^{t+4} \mathbf{1}_{\mathcal{Z}}(x_{t'}, y_{t'}) \leq 1
\]

which ensures that the agent does not remain in an unsafe region for more than one consecutive time steps.

\section{Supplement for the Main Text}
\label{app:suppl}

\subsection{Figures of Automatas.}

We now show automatas for specifications of Ex.~\ref{ex:running}
and Ex.~\ref{ex:conditional}.

\begin{figure*}[t!]
\centering
  \includegraphics[width=1\linewidth]{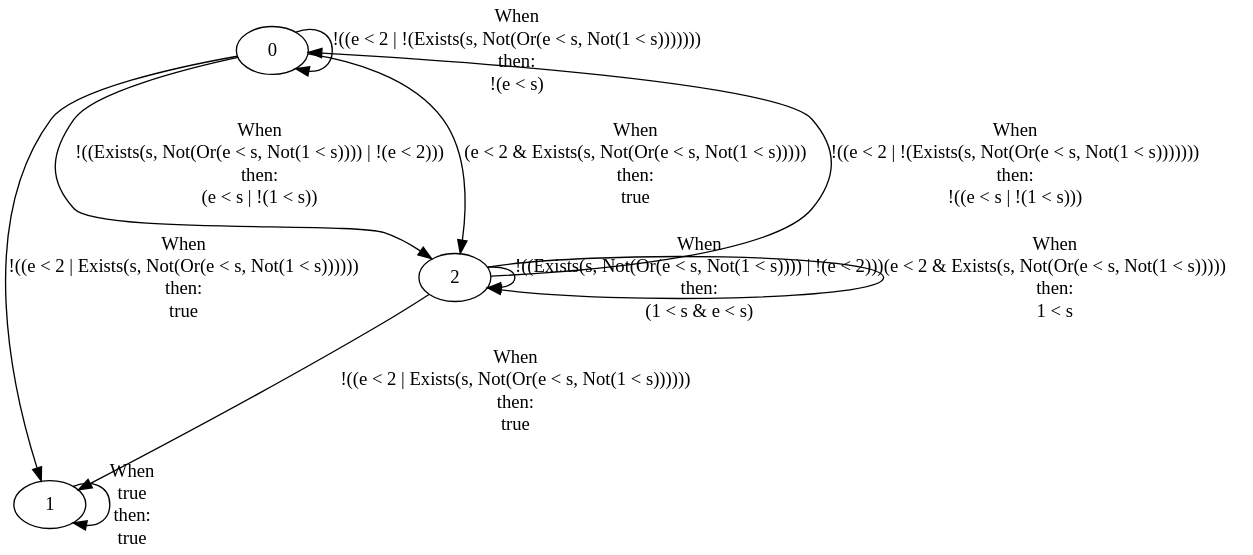}
  \caption{Synthetised strategy for Ex.~\ref{ex:intro}. Note that $x$ of Ex.~\ref{ex:intro} is here $e$ and $y$ is $s$.}
  \label{fig:simpleAutomata}
\end{figure*}

\begin{figure*}[t!]
\centering
  \includegraphics[width=1\linewidth]{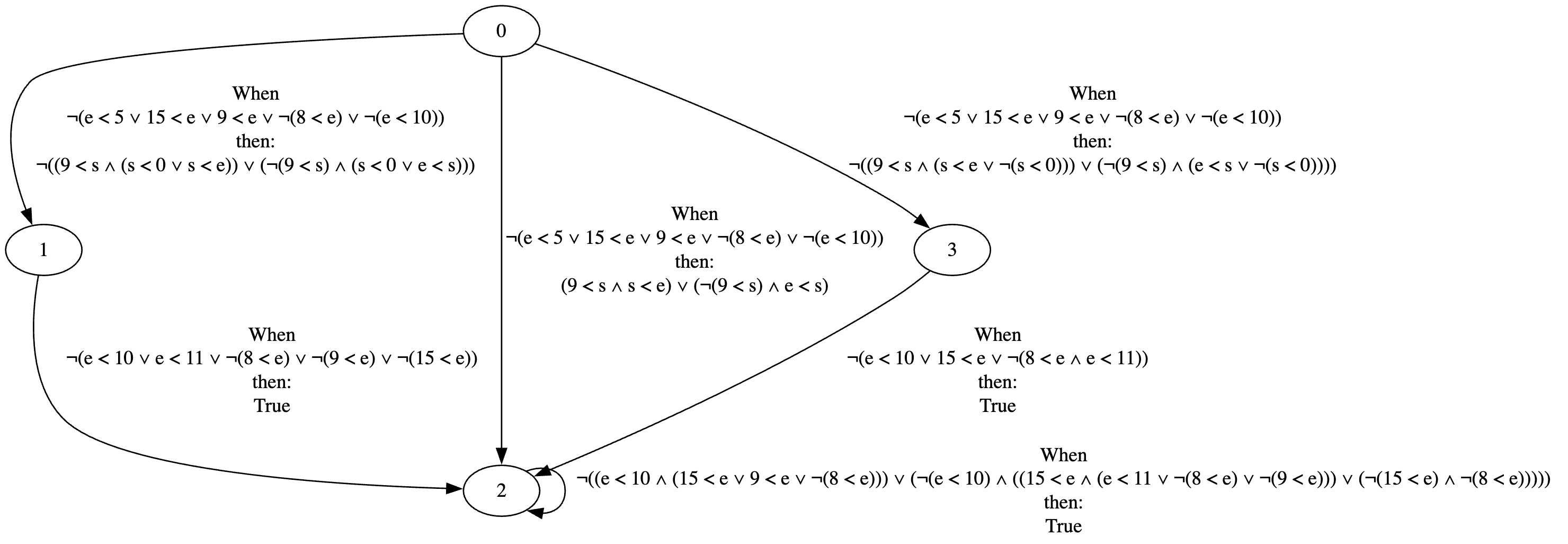}
  \caption{Synthetised strategy for Ex.~\ref{ex:conditional}}
  \label{fig:conditionalAutomata}
\end{figure*}

Moreover, the explainability problem from synthetised
automata quickly becomes worse as we increase the complexity 
of the specifications.
For instance, consider the following slightly more complex specification:

 \begin{align*}
   \varphi_{\mathcal{T}} : \square
   \begin{pmatrix}
     \begin{array}{rcl}
   \textcolor{green}{\Event} (x<10) & \Impl & \Next^{\textcolor{green}{2}}(y>9)
 \\  &\wedge& \\
  (x \geq 10)  &\Impl&  \textcolor{green}{\Event} (y \leq x),
        \\ 
     \end{array}
   \end{pmatrix}
 \end{align*}
 where the green colour denotes the changes with respect to the
 specification of Ex.~\ref{ex:running}.

The strategy for this specification is depicted in Fig.~\ref{fig:complexAutomata}, where
we can see that the difficulty of extracting explanations 
from Mealy machines gets impossible even with reasonable specifications.

\begin{figure*}[t!]
\centering
  \includegraphics[width=1\linewidth]{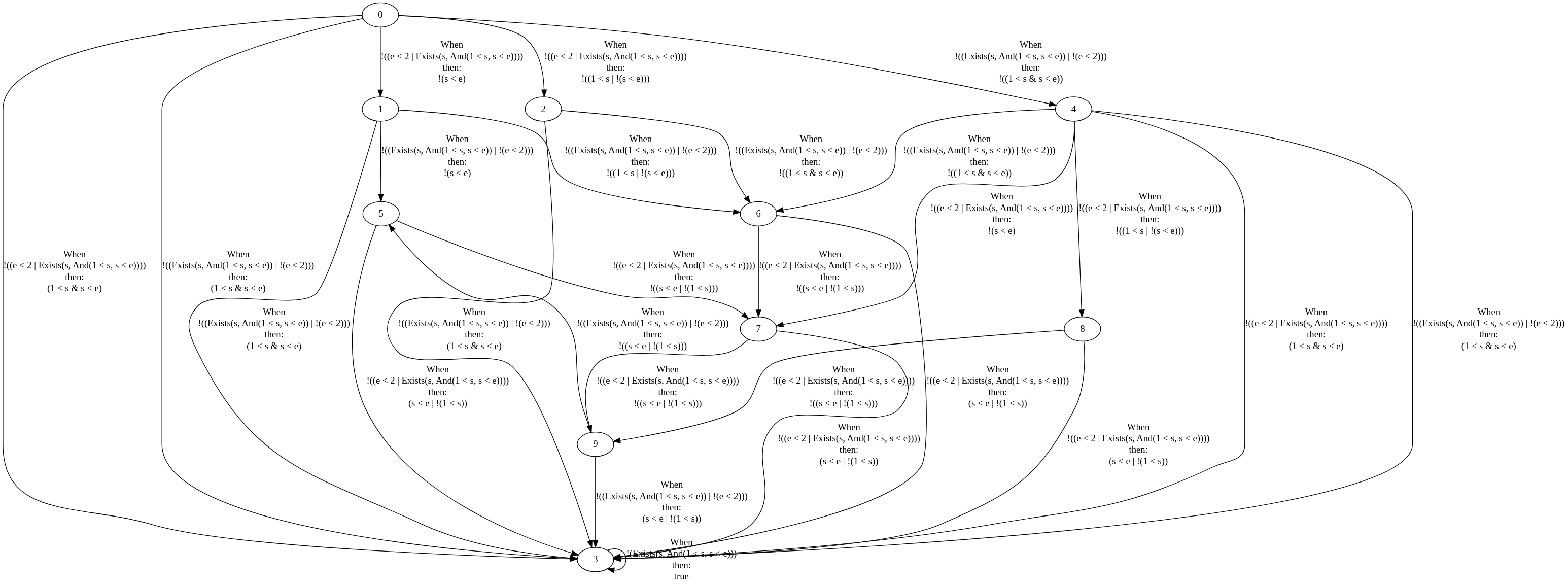}
  \caption{A large automata (note that it is supposed to be unintelligible).}
  \label{fig:complexAutomata}
\end{figure*}

\subsection{Missing Algorithms}

Find below Alg.~\ref{algoQBFPlus}.

\begin{algorithm}
\begin{algorithmic}[1]
  \REQUIRE $\phiB$, $\Max$, $E$, $S$ 
  \STATE $T \gets \emptyset$ 
  \FOR{$n=1 \text{ to } \Max$}
    \STATE $[s_0,\ldots,s_n] \gets \Copies(S,n)$ 
    \STATE $[e_0,\ldots,e_n] \gets \Copies(E,n)$ 
    \STATE $F \gets \Unroll(\phiB,n)$ 
    \STATE $G \gets \texttt{alternations}([s_0,\ldots,s_n],[e_0,\ldots,e_n]). F$  
    \STATE $\phiB^{\text{qbf}} \gets \QElim([s_0,\ldots,s_n],G)$\\
    \IF{$\neg \phiB^{\text{qbf}}$ is SAT}
      \STATE \textbf{return} $(\texttt{true}, \Witness(\phiB^{\text{qbf}}))$
    \ENDIF
  \ENDFOR
  \STATE \textbf{return} \texttt{unreal} 
\caption{Cond. bounded unrealiz. check for LTL.}
 \label{algoQBFPlus}
 \end{algorithmic}
\end{algorithm}


\end{document}